\documentclass{amsproc}
\usepackage{amsthm,graphicx,amsmath,amssymb,amsfonts,framed}
\usepackage{slashed}
\usepackage{graphicx,color}
\newcommand{\norm}[1]{\left\lVert#1\right\rVert}
\newcommand{\abs}[1]{\lvert#1\rvert}

\newcommand{\scl}[2]{\langle#1,#2\rangle}
\newcommand{\suup}[1]{ \underset{#1}{\sup} }

\def\begf{\begin{frame}}
\def\enf{\end{frame}}

\def\begz{\begin{itemize}}

\def\endz{\end{itemize}}

\def\lp{\left(} 
\def\rp{\right)} 
\def\dm{\lp\begin{array}}	
\def\fm{\end{array}\rp}
\def\begf{\begin{frame}}
\def\enf{\end{frame}}

\def\begz{\begin{itemize}}

\def\endz{\end{itemize}}

\def\lp{\left(} 
\def\rp{\right)} 
\def\dm{\lp\begin{array}}	
\def\fm{\end{array}\rp}
\def\m2{M_2 \lp \cc \rp}
\def\m3{M_3 \lp \cc \rp}	
\def\mn{M_n \lp \cc \rp}	
\def\re{\text{Re}}			
\def\im{\text{Im}}
\def\ot{\otimes}

\def\ds{\partial\!\!\!\slash}

\def\cc{{\mathbb{C}}}
\def\C{{\mathbb{C}}}	
\def\rr{{\mathbb{R}}}
\def\R{{\mathbb{R}}}
\def\nn{{\mathbb{N}}}
\def\N{{\mathbb{N}}}
\def\zz{{\mathbb{Z}}} 		
\def\ii{{\mathbb{I}}}
\def\I{{\mathbb{I}}}
\def\hhh{{\mathbb H}}

\def\mm{{\mathcal M}}	
\def\M{{\mathcal M}}		
\def\aa{{\mathcal A}}

\def\A{{\mathcal A}}

\def\hh{{\mathcal H}}

\def\pp{{\mathcal P}}

\def\X{{\mathcal X}}

\def\K{{\mathcal K}}

\def\cinf{C^{\infty}\lp\mm\rp}

\def\lda{\text{Lip}_D(\A)}

\def\xox{{\xi}_x}
\def\yox{{\xi}_y}
\def\xoz{{\zeta}_x}
\def\yoz{{\zeta}_y}
\def\ox{\omega_{\xi}}
\def\oz{\omega_{\zeta}}

\def\ot{\otimes}

\def\xox{{\xi}_x}
\def\yox{{\xi}_y}
\def\xoz{{\zeta}_x}
\def\yoz{{\zeta}_y}

\def\ox{\omega_{\xi}}
\def\oz{\omega_{\zeta}}
\def\xo0{\omega^0_x}
\def\yo0{\omega^0_y}

\def\ou{{\omega_{1}}}
\def\od{{\omega_{2}}}

\def\ou{{\omega_{1}}}

\def\oc{{\omega_{c}}}

\def\akom{\alpha_\kappa\omega_m}
\def\akton{\alpha_{\tilde\kappa}\omega_n}

\def\xo0{x_\omega^0}
\def\yo0{y_\omega^0}

\def\aea{\alpha_e(\aa)}

\def\pa{{\mathcal P}(\aa)}
\def\sa{{\mathcal S}(\aa)}

\def\fm{\Phi(x^\mu)}

\def\dm{\partial_\mu}

\def\X{{\mathcal X}}

\def\dmm{\left(\begin{array}}
\def\fmm{\end{array}\right)}

\newcommand{\HH}{\mathcal{H}}
\newcommand{\de}{\mathrm{d}}

\newtheorem{theorem}{Theorem}[section]

\newtheorem{cor}[theorem]{Corollary}
 \newtheorem{defi}[theorem]{Definition}
 
\newtheorem{lem}[theorem]{Lemma}
\newtheorem{prop}[theorem]{Proposition}
\newtheorem{thm}[theorem]{Theorem}
\theoremstyle{definition}
\newtheorem{definition}[theorem]{Definition}

\theoremstyle{remark}

\numberwithin{equation}{section}

\begin{document}

\title[Distances in NCG: from Monge to Higgs]{From Monge to Higgs:\\[5pt] a survey of distance computations
in \\[3pt]noncommutative geometry}

\author[P. Martinetti]{Pierre Martinetti}
\address{
Dipartimento di Matematica, Universit\`a di Genova}
\curraddr{}
\email{martinetti@dima.unige.it}


\date{11 January 2015}

\begin{abstract} This is a review of explicit computations of Connes
  distance in noncommutative geometry, covering finite dimensional
  spectral triples, almost-commutative geometries, and spectral triples
  on the algebra of compact operators. Several applications to physics
  are covered, like the metric interpretation of the Higgs field, and
  the comparison of Connes distance with the minimal length that
  emerges in various models of quantum spacetime. Links with other areas of
  mathematics are studied, in particular the horizontal
  distance in sub-Riemannian
  geometry.  The interpretation of Connes
  distance as a noncommutative version of the Monge-Kantorovich metric
  in optimal transport is also discussed.  
\end{abstract}

\maketitle

\tableofcontents

\noindent \emph{Proceedings of the workshop ``Noncommutative Geometry and
  Optimal Transport'', Besan\c{c}on november 2014.}
\newpage
\section{Introduction}
The distance formula in noncommutative geometry has been introduced
by Connes at the end of the 80's \cite{Connes:1989fk}. Given a
so-called \emph{spectral triple} $(\A, \HH, D)$, that is  an involutive algebra $\A$ acting on a
Hilbert space  ${\mathcal H}$ via a representation $\pi$, and an operator $D$ on $\mathcal H$ such
that the commutator $[D, \pi(a)]$ is bounded for any $a$ in $\A$, one
defines on the space $\sa$ of states of $\A$ 
the (possibly infinite) distance 
\begin{equation}
  \label{eq:-1}
  d(\varphi, \varphi') =\sup_{a\in\A}\left\{|\varphi(a) -\varphi'(a)|,\;
    ||[D,\pi(a)]||\leq 1\right\} \quad \forall \varphi, \varphi'\in\sa.
\end{equation}
For $\A=C_0^\infty(\M)$ the commutative algebra of smooth functions vanishing at infinity on a
locally compact and complete manifold $\M$, acting on the Hilbert
space $\HH$
of square integrable differential forms and $D$ the signature
operator, this distance computed between pure states gives back the
geodesic distance on $\M$.  In this sense,  eq.~\eqref{eq:-1} is a generalization of Riemannian geodesic distance that makes
sense in a noncommutative setting, and provides an original tool to study the geometry of the space of states
on an algebra. Besides its mathematical inte\-rest, Connes distance also has an intriguing echo in physics,
for it yields a metric
interpretation for the Higgs field \cite{Connes:1996fu}, the missing piece of the
\emph{Standard Model of Fundamental Interactions} recently discovered
by the Large Hadronic Collider at CERN.

In the 90's, Rieffel \cite{Rieffel:1999wq} noticed that \eqref{eq:-1} was a noncommutative version of the Wasserstein distance of
order $1$ in the theory of optimal transport (the modern
version of Monge \emph{d\'eblais et
  remblais} problem). More exactly, this is a
noncommutative generalization of Kantorovich dual formula of the
Wasserstein distance \cite{Kan42}. Formula \eqref{eq:-1}, which we call
\emph{spectral distance} in this survey, thus offers an unexpected
connection 
between an ancient mathematical problem and the most recent discovery
in high energy physics.

The meaning of this connection is far from clear. Yet, Rieffel's
observation suggests that the spectral distance may provide an interesting starting point for a theory of optimal transport in
noncommutative geometry, as well as  a possible interpretation of
the Higgs field as a cost function on spacetime. More specifically, one may wonder
\begin{itemize}
\item  What remains of the duality
Wasserstein (minimizing a cost)/Kantorovich (maximizing a profit) in the noncommutative setting ?  Is there some
``noncommutative cost'' that one is minimizing while computing the
supremum in the distance formula ?  

\item May the noncommutative geometry point of view on the Wasserstein distance help to
solve some problems in optimal transport ? Vice-versa, can one use
results of optimal transport to address relevant issues in
noncommutative geometry ? 

\item Is such a generalization of the Riemannian distance truly interesting
  for physics~?
\end{itemize}

These questions were at the origin of
the mini-workshop \emph{Optimal transport and noncommutative geometry}
hold in Besan\c{c}on in november 2014, and whose present text is part
of the proceedings. Both optimal transport and
noncommutative geometry are 
active areas of research, but with little intersection. In addition, the metric aspect of noncommutative
geometry is a part of the theory that has been relatively little
studied so far  \cite{Connes:1992bc}. Nevertheless several results -
including explicit computations - have
been obtained in the recent years, and links with other areas of geometry (like
sub-Riemannian geometry) have been discovered.  

 This survey aims at providing an account of the metric aspect of
noncommutative geometry, readable by non experts.  The questions
listed above will serve as a guideline (they are
discussed in a more systematic way in the last section of the paper),
but our point of view is rather to emphasize explicit
calculations of the spectral distance, starting
with commutative examples and going further in noncommutativity:
finite dimensional algebras, matrix valued functions on a manifold,
compact operators.  We omit the proof (that can be found in the indicated bibliography) and
stress for each example some application in physics, or some relation
with other part of mathematics. 
\smallskip

More precisely,  after some generalities in section~\ref{sec:general} where we introduce formula
\eqref{eq:-1} and discuss some of its properties, we begin our survey
in section \ref{section:finite} with
finite dimensional spectral triples. This is essentially a review of \cite{Iochum:2001fv}
with some slight generalizations to non-pure states. Depending on the finite dimensional algebra
being commutative or not, one deals with distance on a graphs
(\S\ref{sec:discretespaces} and \S \ref{sec:graph}) or on projective
spaces (\S\ref{sec:proj}), like the sphere (\S\ref{section:ball}).  In
section \ref{section:almostcg} we consider products of spectral
triples. After general properties in \S\ref{section:pythagoras} mainly
taken from \cite{DAndrea:2012fkpm}, we
focus on \emph{almost commutative geometries} in \S \ref{section:almost}, that is the product of a manifold by a finite dimensional spectral
triple. This is in this context that the Higgs field acquires a
metric interpretation \cite{Chamseddine:1996kx,Connes:1996fu}
as the component of the metric in a discrete internal dimension
\cite{Martinetti:2002ij}, as explained in \S\ref{higgs}. Section \ref{subriemannian} is entirely
devoted to the relation between
almost-commutative geometry and sub-Riemannian geometry. As recalled
in \S\ref{sec:subriem}, this relation
has
been pointed out in \cite{Connes:1996fu} but fully studied in
\cite{Martinetti:2006db, Martinetti:2008hl}: formula \eqref{eq:-1} yields a
(possibly infinite) distance on the bundle
$\pa$ of pure states of
the algebra of matrix-valued functions on a manifold, which is finite between certain classes of
 leaves of the horizontal foliation of $\pa$. This is in contrast
with the horizontal distance which, by definition, is infinite between
horizontal leaves. The difference between the spectral and the
horizontal distances is governed by the holonomy associated to a Dirac
operator of an almost commutative geometry (\S\ref{sec:obstruction}). The
computation of the distances is worked out in details for a simple
example (bundle on a circle) in \S\ref{sec:counterexample}-\ref{sec:fiber}
In section \ref{sec:compact} we consider a truly noncommutative example,
that is a spectral triple on the algebra of compact operators. We view
the latter first as the algebra of the Moyal plane in \S\ref{section:Moyal}, then as the algebra
describing some models of quantum spacetimes in \S\ref{DFR}. 
In the last section of the paper, we discuss various problems, in particular what
could play the role of geodesics in a noncommutative framework
(\S\ref{points} and \S\ref{sec:geodesics}), and
how to export Kantorovich duality to the noncommutative side (\S\ref{sectionncg}).
\smallskip

Although we try to cover a wide range of examples, this survey is not
exhaustive.
For a state of the art on the topological aspect of
metric noncommutative geometry, we invite the reader to see the nice
review of Latr\'emoli\`ere in this volume
\cite{Latremoliere:2015ab}, or \cite{Mesland:2015aa} for an approach
oriented towards $KK$-theory. Among the subjects that are not treated
here, let us mention applications to dynamical
systems \cite{Bellissard:2010fk},  fractals (see e.g. \cite{Christensen:2011fk,
  Christensen:2006fk}), as well as the pseudo-Riemannian case, e.g. in \cite{Moretti:2003zw} and \cite{Franco:2010fk}.
\medskip

The notations are collected at the end of the text, before the bibliography.

\section{Distances on the space of states of an algebra}
\label{sec:general}
A state $\varphi$ of a complex $C^*$-algebra $\A$ is a linear
application from $\A$ to $\C$ which is positive (any positive element $a^*a$
of $\A$ has image a non-negative real number)
and of norm $1$. A similar notion exists
for real algebras, although one should be careful that
selfadjointness, $\varphi(x^*) = \bar\varphi(x)$, does not follow from
positivity, as it does for unital complex algebras \cite{Goodearl:1982fk}.

For any $C^*$-algebra the set of states $\sa$ is convex, and even compact in the weak-$*$ topology in case the algebra is unital. The extremal
points are the pure states $\pa$. By Gelfand theorem, for
$\A=C_0(\X)$ the algebra of continuous functions vanishing at infinity
on a locally compact topological space $\X$, the pure states are in
$1$-to-$1$ correspondence with the points of $x$, viewed as the evaluation
\begin{equation}
  \label{eq:02}
  \delta_x(f) := f(x) \quad \forall x\in\X, f\in C_0(\X).
\end{equation}
Taking as a rough definition of noncommutative geometry a ``space
whose algebra of functions $\A$ is non-commutative'', pure states of $\A$ thus appear as
natural candidates to play the role of points in a noncommutative
framework. One may prefer to focus on classes of irreducible
representations rather than on pure states; this is discussed
in \S \ref{points}.

\subsection{Commutative case: the Monge-Kantorovich distance}
\label{subsec:MK}
In the commutative case $\A= C_0(\X)$, a distance in the
space of states $\sa$ is provided by optimal transport. Namely, given a function
$c:\X\times\X \to \R$ called the \emph{cost function}, the optimal
transport between two probability measures $\mu_1, \mu_2$ on $\X$ is 
\begin{equation}
  \label{eq:3}
  W(\mu_1, \mu_2) :=\inf_{\rho} \int_{\X\times \X} c(x,y) \, d\rho(x, y)
\end{equation}
where the infimum runs on all the measure $\rho$ on $\X\times\X$ with
marginals $\mu_1, \mu_2$. When the cost function $c$ is a distance,
then $W(\mu_1, \mu_2)$ is a distance on the space of probability
measures on $\X$, called the \emph{Wasserstein} or the \emph{Monge-Kantorovich}
distance of order $1$. One obtains a distance on the space of states
noticing that  any probability measure $\mu$
 defines a state 
\begin{equation}
  \label{eq:4}
  \varphi(f) = \int_\X f(x)\, d\mu(x),
\end{equation}
and any state comes in this way.

The Monge-Kantorovich distance is important
in probability theory because  the convergence in $W$ always implies the weak$^*$
 convergence (with convergence of moments). For $\X$ compact, $W$ actually metrizes the weak* topology on probability
 measures. This is not the only distance to make it, but according
 to Villani \cite[p. 97]{Villani:2009tp}  this is the most convenient one.

A similar definition exists for any order
$p\in\N^*$, by considering instead of  \eqref{eq:3}
\begin{equation}
  \label{eq:138}
 W_p(\mu_1, \mu_2) :=\inf_{\rho} \left(\int_{\X\times \X} c^p(x,y) \,
   d\rho(x, y)\right)^{\frac 1p}.
\end{equation}
Nevertheless, in
this paper we will mostly consider the distance of order $1$,
because in this particular case there exists a dual formulation which
makes sense in a noncommutative context. Indeed, 
Kantorovich showed \cite{Kan42} that $W$ can be equivalently written as
\begin{equation}
  \label{eq:05}
    W(\mu_1, \mu_2):=\sup_{||f||_\text{Lip}\leq 1} \int_\X f d\mu_1 -
    \int_\X f d\mu_2 
\end{equation}
where the supremum runs on all real functions which are Lipschitz with
respect to the cost, that is 
\begin{equation}
  \label{eq:6}
  |f(x) - f(y)| \leq c(x,y) \quad\forall x,y \in \X. 
\end{equation}
As explained in \S \ref{subsection:spectraldist} below, for ${\mathcal X} = \M$
a Riemannian complete manifold,  the dual form \eqref{eq:05} of
the Wasserstein distance  coincides with the spectral distance
\eqref{eq:-1} for $\A=C^\infty_0(\M)$ acting on the space of
differential forms and $D$ the signature operator. 
\smallskip

Before entering the details, let us stress why the dual formulation of
Kantorovich may be of  interest for physics. Computed between pure states, $W(\delta_x, \delta_y)= c(x, y)$ gives
back the cost function. In
particular on a Riemannian manifold $\M$, taking as cost the geodesic
distance, the Wasserstein distance
\eqref{eq:05} between pure states provides an algebraic formulation of
the geodesic distance in terms of supremum, in contrast with the usual
definition as the infimum on the length of all paths between $x$ and $y$. This view on the geodesic distance does not rely on any
notion ill defined in the quantum context, such as points or
path between points. It
only involves algebraic tools, 
typical from quantum physics.
\smallskip 

Possible noncommutative generalizations of $W_p$ for $p\geq 2$ are
discussed in \S \ref{sectionncg}.

\subsection{Noncommutative case: Connes spectral distance}
\label{subsection:spectraldist}

A distance on the space $\sa$ of states of a non-necessarily commutative $C^*$-algebra $\A$
has been introduced by Connes at the end of the 80's \cite{Connes:1989fk} in the framework of noncommutative geometry. 

Assuming $\A$ acts on an Hilbert space $\HH$,   then given an operator
$D$ on $\HH$, one
associates to any pair of states $\varphi, \varphi'\in\sa$ 
the quantity
\begin{equation}
  \label{eq:1}
  d(\varphi, \varphi') :=\sup_{a\in \text{Lip}_D(\A)} |\varphi(a) -\varphi'(a)|
\end{equation}
where the $D$-Lipschitz ball of $\A$ is the subset of $\A$ defined as
\begin{equation}
  \label{eq:14}
  L_D(\A):= \left\{a \in \A, ||[D,a]||\leq 1\right\},
\end{equation}
where the norm ia the operator norm on $\HH$. In \eqref{eq:14} as well as most of the
time in the rest of the paper, we omit the symbol of representation and we identify an element $a$ of $\A$ with
    its representation $\pi(a)$ as bounded operator on $\HH$. In these conditions
    \eqref{eq:1} coincides with \eqref{eq:-1}.

 Eq. \eqref{eq:1} is obviously invariant under the exchange of $\varphi$
and $\varphi'$, and is zero if $\varphi =\varphi'$. The triangle
inequality is easy to check. For two states $\varphi, \varphi'$
that are equal everywhere but on some element $a_\infty$  such that
$[D, a_\infty]$ is unbounded, one has $d(\varphi, \varphi')=0$
although $\varphi\neq \varphi'$. To avoid this, one requires $[D,a]$ be bounded for any $a\in\A$. Then \eqref{eq:1} defines a
distance (possibly infinite) on $\sa$. Asking $(D-\lambda \I)^{-1}$ to
be
compact for any $\lambda$ in the resolvent set of $D$ (in case $\A$ is
unital, this means that $D$ has compact resolvent), the set $(\A, \HH,
D)$ is called a \emph{spectral triple} (and $D$ a \emph{Dirac operator}). Hence the name \emph{spectral
  distance} to denote \eqref{eq:1}. 

For $\A=C_0^\infty(\M)$ the algebra of smooth functions vanishing at infinity on a
locally compact complete Riemannian manifold $\M$, with multiplicative
representation on the Hilbert space $\HH= L^2(\M, \wedge)$ of square
integrable differential forms,
\begin{equation}
     \label{eq:13}
     (f\psi)(x) = f(x)\psi(x) \quad \forall x\in\M, \psi\in L^2(M, \wedge),
   \end{equation}
 and $D=d+ d^\dagger$ the signature
operator ($d$ is the exterior  derivative, $d^\dagger$ its adjoint)
then the spectral distance \eqref{eq:1} computed between pure states gives back the
geodesic distance on $\M$,
\begin{equation}
  \label{eq:15}
  d(\delta_x, \delta_y) = d_\text{geo}(x, y).
\end{equation}
A similar result is obtained, in case $\M$ is a spin manifold, with
$\HH= L^2(\M, S)$ the Hilbert space of square integrable spinors and 
  \begin{equation}
 D= \ds :=-i\gamma^\mu(\partial_\mu + \omega_\mu)
 \label{eq:Diracoperateur}
 \end{equation}
the usual Dirac
 operator, with  $\omega_\mu$ the spin connection and $\gamma^\mu$ the
 Dirac matrices satisfying
 \begin{equation}
\gamma^\mu\gamma^\nu +
 \gamma^\nu\gamma^\mu = 2g^{\mu\nu}\I
\label{eq:142}
 \end{equation}
 where $g^{\mu\nu}$ the
 Riemannian metric on $\M$. 

Furthermore, in \cite{Rieffel:1999ec} Rieffel noticed that for $\M$
compact and for any state (pure or not), formula \eqref{eq:1} was nothing but Kantorovich dual
formulation \eqref{eq:05} of the Wasserstein distance. This is because the
norm of the commutator $[d+d^\dagger, f]$ (or $[\ds, f]$ in case the
spin structure is taken into account) is nothing but the Lipschitz
norm of $f$.  We show in
\cite{dAndrea:2009xr} that this remains true 
for a locally compact manifold as soon as it is complete (the latter
condition guarantees 
that looking for the supremum on $C^\infty_0(\M)$ or on Lipschitz
functions is equivalent). 
\bigskip

Therefore, Connes spectral distance appears as a 
generalization of the Wasserstein distance. More precisely, it
provides a formulation of Kantorovich dual
formula which makes sense also in a noncommutative context. Whether there exists a noncommutative version of the initial
definition \eqref{eq:3} of the Wasserstein distance as an infimum is an open
question, discussed in section \ref{sectionncg}.

Let us stress that Connes formula makes sense in a wider context: one may look for the
supremum on the Lipschitz  ball $L(a)\leq 1$ for any seminorm $L$ on
$\A$, non necessarily coming from the commutator with an operator. One does not even need to work with an algebra: states and Lipschitz
seminorms makes sense for \emph{ordered unit spaces} (see
\cite[\S 11]{Rieffel:1999ec} for an extended discussion on that matter). This flexibility is useful when one
focuses on topological aspects of the distance (for instance under
 which conditions does \eqref{eq:1}
 metrize the weak$^*$ topology on $\sa$~?~\cite{Latremoliere:2015ab}). In this review, we adopt the point
 of view that spectral triples provide algebras and operators $D$ - hence seminorms $L_D$ - 
 that are relevant for physics as well as for other aspect of mathematics,
 offering thus various examples where
 the explicit computation of the spectral distance is worth
 undertaken.

It is also worth mentioning 
that by adding more conditions on $\A$, $\HH$ and $D$, one is able to fully characterize a
Riemannian closed (spin) manifold $\M$ as a spectral triple
$(\A, \HH, D)$ where $\A$ is commutative
\cite{connesreconstruct}. Focusing only on the metric aspect, one may as well be interested in characterizing a metric space in terms
of algebraic datas, without the need of any smooth
structure. A good reference on this topic is \cite{Weaver:1999aa}. A
general reference on the algebraic way of characterizing a smooth
manifold is \cite{Netsruev:2003aa}.

\subsection{Isometries \& projections}
Before making explicit computations of the distance, let us list
various definitions and easy but useful general results. In all this section, $d$
denotes the spectral distance \eqref{eq:1} associated to an arbitrary
spectral triple $(\A, \HH, D)$.

\begin{defi}
\label{optimelement}
  We call ``optimal element for a pair of states $\varphi, \varphi'\in\sa$'' any element $a$ in $L_D(\A)$ such that 
  \begin{equation}
|\varphi(a)  - \varphi'(a)|  = d(\varphi, \varphi')
\label{eq:121}
\end{equation}
or, in case the supremum is not reached, any sequence $\left\{a_n\in L_D(\A)\right\}$ such that
  \begin{equation}
    \label{eq:122}
 \lim_{n\to\infty} |\varphi(a_n)  - \varphi'(a_n)| = d(\varphi, \varphi').
  \end{equation}
\end{defi}
\begin{lem}\cite[Lem. 1]{Iochum:2001fv}
  The supremum in  \eqref{eq:1} can be searched equivalently on
  selfadjoint elements of $\A$. In case $\A$ is unital, the supremum
  can be equivalently searched on positive elements.
\end{lem}

 We call isometry of the state
space an application $\alpha: \sa \to \sa$ such that 
\begin{equation}
  d(\varphi, \varphi') = d(\alpha(\varphi), \alpha(\varphi'))
  \quad\forall \varphi, \varphi'\in \sa.
\label{eq:46}
\end{equation}
A class of isometries particularly useful for explicit
computations are the lift to states of inner automorphisms of $\A$, that
is\begin{equation}
  \label{eq:47}
  \alpha_u(\varphi) := \varphi\circ\alpha_u
\end{equation}
where $\alpha_u:=\text{Ad} \, u$ for some unitary  $u\in\A$.
\begin{lem} 
\label{selfadjointlem}
Let $u$ be a unitary element in $\A$ that commutes with $D$, then
$\alpha_u$ is an isometry of $\sa$. Namely 
\begin{equation}
d(\varphi, \varphi') = d(\varphi\circ\alpha_u, \varphi'\circ\alpha_u)
\quad \forall\varphi, \varphi'\in\sa.
\label{eq:144}
\end{equation}
\end{lem}
\noindent The proof is easy an can be found e.g. in
\cite[Prop. 1.29]{Martinetti:2001fk}. The result is also valid for
some operator  $u$ in ${\mathcal B}(\HH)$
that is not necessarily the representation of a unitary element of
$\A$. In this case one should consider only the states whose domain contains $u$, that is
such that $\varphi\circ \alpha_u$ and $\varphi'\circ \alpha_u$ make sense.
 \medskip

Other useful applications are projections, that sometimes permit to reduce the
search for the supremum in \eqref{eq:1} to subsets of $\A$ more
tractable than $\A$ itself. 
\begin{defi}
   The projection of a spectral triple $(\aa, \hh, D, \pi)$ ($\pi$ is
  the representation of $\A$ on $\HH$) by a projection
  $e=e^*=e^2\in{\mathcal B}(\HH)$ is the triple
  \begin{equation}
    \label{tripletrestreint}
    \aa_e:= \aea,\quad \hh_e:= e\hh,\quad
    D_e:= eDe\big|_{\hh_e},
  \end{equation}
where $\alpha_e(a):= e\pi(a)e$ for any $a\in\A$. \end{defi}
\noindent The projected triple $(\aa_e,\, \hh_e,\, D_e)$ may not be a spectral
triple since  $\aa_e$ may not be an algebra,  for instance when
$e\notin\pi(\A)$. Nevertheless, when $\A$ is unital the set
\begin{equation}
  \label{eq:50}
  \aa_e^{\text{sa}}:=\left\{ e\pi(a)e, a=a^*\in\A\right\}
\end{equation}
of selfadjoint elements of $\aa_e$ is an ordered unit
space. Therefore, as explained at the end
of \S \ref{subsection:spectraldist},
the
notion of states of $\aa_e^{sa}$ - and by extension of $\A_e$ - still makes sense, with an obvious map from ${\mathcal S}(\aa_e)$ to $\sa$,
\begin{equation}
\varphi\to \varphi\circ\alpha_e.
\label{eq:51}
\end{equation} 
Given $\varphi, \varphi'\in {\mathcal S}(\A_e)$, we  still
call \emph{spectral distance} the quantity
\begin{equation}
  \label{eq:48}
  d_e(\varphi, \varphi'):= \sup_{\text{Lip}_D(\aa_e)}
  |\varphi(a) - \varphi'(a)|.
\end{equation}
\begin{lem}
\cite [\label{projectionlem}Lem. 1]{Martinetti:2002ij}
Let $(\A, \HH, D)$ be a unital spectral triple and $e$ a projection in $\HH$
that commutes with $D$, then for any states $\varphi, \varphi'$ of $\aa_e$ one has
\begin{equation}
d_e(\varphi,\varphi')=d(\varphi\circ\alpha_e,\varphi'\circ\alpha_e).
\label{eq:49}
\end{equation}
\end{lem}

Said differently, a projection that commutes with the Dirac
operator behaves like an isometry. The difference between
\eqref{eq:49}  and \eqref{eq:46} is that in \eqref{eq:49} the set of elements on
which the supremum is searched is smaller on the l.h.s. than on the
r.h.s. 
Notice also that the application \eqref{eq:51} has no reason to be
surjective.

\subsection{Connected components}

Given a spectral triple $(\A, \HH, D)$,
we denote the set of states at finite spectral distance from  a state $\varphi\in\sa$ by
\begin{equation}
\label{con} \text{Con}(\varphi):= \{\varphi'\in \sa;\; d(\varphi, \varphi') <
+\infty\}.
\end{equation}
The
notation is justified because this set
coincides with the connected component of $\varphi$ in $\sa$ for the
topology metrized by the spectral distance (see \cite
[Def. 2.1]{dAndrea:2009xr}).
\begin{prop}
\label{convex}
For any $\varphi\in\sa$, the set  $\text{Con}(\varphi)$ is convex.
\end{prop}
\begin{proof}
For any $\varphi_0, \varphi_1\in \text{Con}(\varphi)$ and $s\in[0,1]$, denote
\begin{equation}
\varphi_s:= s\,\varphi_0 + (1-s)\,\varphi_1.
\label{eq:57}
\end{equation}
One easily checks that 
\begin{equation}
  \label{eq:58}
  d(\varphi_s, \varphi_t) = |s-t|\, d(\varphi_0, \varphi_1) \quad \forall s,t\in\R.
\end{equation}
By the triangle inequality, $d(\varphi_0, \varphi_1)$ is 
finite. Thus the same is true for   $d(\varphi_s, \varphi_t)$. In particular $d(\varphi_0, \varphi_s)$ is
finite, so again by the triangle inequality
$d(\varphi, \varphi_s)$ is finite for any  $s\in[0,1].$ Hence 
$\varphi_s\in \text{Con}(\varphi)$, showing the later is convex.
\end{proof}
 Restricting the connected component of a state to pure
states, by prop. \ref{convex} one obtains a set whose convex hull is still in the connected component,
\begin{equation}
\overline{\text{Con}(\varphi) \cap \pa} \subset \text{Con}(\varphi).
\label{eq:52}
\end{equation}
But at this point nothing guarantees that $\text{Con}(\varphi)$ is the
convex hull of its restriction to pure states. We come
back to this point in section \ref{sectionncg}.
\medskip

The following lemma is useful to characterize the connected components. 
 \begin{lem}
\label{infinitelem}
For any two states $\varphi,
\varphi'$ of $\A$,  the distance $d(\varphi, \varphi')$ is infinite if
and only if there exists a sequence $a_n\in\A$ such that 
\begin{equation}
  \label{eq:62}
  \lim_{n\to\infty}  L_D(a_n)=0  \;\text{ and }\;  \lim_{n\to\infty}
  \varphi(a_n)  - \varphi'(a_n) = \infty.
\end{equation}
In particular $d(\varphi, \varphi')$ is infinite as soon as there
exists an element $a\in \A$ such that 
\begin{equation}
  \label{eq:64}
 L_D(a) = 0 \; \text{ and } \; \varphi(a) \neq \varphi'(a).
\end{equation}
 \end{lem}
 \begin{proof}
   The proof that the non-finiteness of the distance is equivalent to \eqref{eq:62} is easy and can be found for instance
   in \cite[Lemma 1]{Martinetti:2006db}. The second statement follows
   by considering $a_n := n a, \; n\in \N$.  \end{proof} 

In the finite dimensional case, there are stronger results.
\begin{lem}
\label{infinitelemfinitedim}
For a spectral triple with finite dimensional $\A$ and $\HH$,  the distance between two states $\varphi$, $\varphi'$ is finite if and
only if 
\begin{equation}
  \label{eq:65}
  \varphi(a) = \varphi'(a) \quad \forall a\in \text{Ker}\;L_D.
\end{equation}
In particular, for $\A$ unital the distance is finite on the whole space of states if
and only if 
 \begin{equation}
   \label{eq:60}
  \text{Ker} \, L_D = \left\{ \lambda \I, \, \lambda\in\C\right\}.
 \end{equation}
 \end{lem}
\begin{proof}
For the first statement, by lemma \ref{infinitelem}
   one just needs to show that
   \begin{equation}
d(\varphi, \varphi')=\infty \Longrightarrow \exists\,  a\in\text{Ker}\,
L_D \text{  such that } \varphi(a)\neq
   \varphi'(a).
\label{eq:67}
   \end{equation}
Let us thus assume $d(\varphi, \varphi')$ is infinite. This means
there exists a sequence $a_n\in \A$ satisfying \eqref{eq:62}.
By hypothesis $\HH$ is isomorphic to $\C^N$ for some $N\in\N$ and $\A$
is a subalgebra of $M_N(\C)$. The kernel of $L_D$ is a vector subspace
of $M_N(\C)$. Let $K^\perp$ denote its orthogonal complement in $M_N(\C)$
and
\begin{equation}
\tilde \A :=  \A \cap K^\perp\simeq \A\slash \text{Ker} \, L_D.
\label{eq:70}
\end{equation}
Any $a_n$ decomposes in a unique way as
\begin{equation}
  \label{eq:69bbis}
  a_n = \tilde a_n + k_n
\end{equation}
where $k_n\in \text{Ker} \, L_D$  and $\tilde a_n\in \tilde \A$.  
On $\tilde\A$, the seminorm $L_D$ is actually a norm. Moreover, since
$L_D(a_n) = L_D(\tilde a_n)$ for
any $n$, by \eqref{eq:62} one gets
\begin{equation}
\lim_{n\to\infty} L_D (\tilde a_n) =
\lim_{n\to\infty} L_D(a_n) =0,
\label{eq:55}
\end{equation}
and all the norms on a finite
dimensional vector are equivalent, so that $\tilde a_n$ tends to
zero in the $C^*$-norm of $M_N(\C)$. Since states are continuous, this means
\begin{equation}
  \label{eq:340bis}
  \lim_{n\to\infty} \varphi(\tilde a_n) - \varphi'(\tilde a_n) = 0,
\end{equation}
hence
\begin{equation}
  \label{eq:34}
   \lim_{n\to\infty} \varphi(a_n) - \varphi'(a_n) = \lim_{n\to\infty} \varphi(k_n) - \varphi'(k_n),
\end{equation}
which is infinite by \eqref{eq:62}. This cannot be true if
\eqref{eq:65} holds true, since the r.h.s. of \eqref{eq:34} would be
zero. Therefore, for a finite dimensional spectral triple the
non-finiteness of $d(\varphi, \varphi')$ implies that $\varphi$ and $\varphi'$ do not
coincide on $\text{Ker}\, L_D$. 

The second statement follows by noticing that for any element $a\neq
\I$, there exist at least two states $\varphi, \varphi'$ that do not
take the same value on $a$. Indeed, given any non-zero $a\in\A$, there
exists at least one state $\varphi$ such that $\varphi(a)\neq 0$
\cite[Theo. 4.3.4]{Kadison1983}. Assume that $\varphi(a)\neq 1$. Then the state
\begin{equation}
  \label{eq:56}
  \varphi' := \frac 12 \varphi + \frac 12 \varphi_0
\end{equation}
where $\varphi_0$ is the state that takes value $1$ on each
$a\in\A$, is such that $\varphi'(a)\neq \varphi(a)$.
If $\varphi(a)=1$, then again by \cite[theo. 4.3.4]{Kadison1983} there
exists at least a state $\varphi'$ such that $\varphi'(a-\I)\neq 0$,
that is $\varphi'(a)\neq 1$.
 \end{proof}
In a wider context (i.e. not necessarily finite dimensional and with a
seminorm not necessarily coming from the commutator with a Dirac-like
operator), condition \eqref{eq:60} is one of the the requirements of what Rieffel called a
\emph{Lip-norm} \cite{Rieffel:1999wq, rieffel2003}, that is a seminorm $L_D$ such that \eqref{eq:1}
metrizes the weak$^*$ topology. For a state of the art of the topological aspect of the
spectral distance, we invite the reader to see the extensive contribution of
Latr\'emoli\`ere in the present volume \cite{Latremoliere:2015ab}.

\section{Finite dimensional algebras}
\label{section:finite}

To begin our survey of explicit computations of the spectral distance \eqref{eq:1},
let us consider finite dimensional (complex) $C^*$-algebras, that is 
finite sums of matrix algebras,
\begin{equation}
  \label{eq:25}
  \A = \bigoplus_{i=1}^N M_{n_i}(\C)
\end{equation}
where $n_i\in\N$ for any $1\leq i\leq N$. We begin by commutative
examples $\A=\C^N$ in \S\ref{sec:discretespaces} and \S\ref{sec:graph},
then we study matrix algebras in \S\ref{sec:proj} and \S\ref{section:ball}.

\subsection{Discrete spaces}
\label{sec:discretespaces}

The simplest case, that is $\A=\C^2$, is instructive although
it is commutative and elementary. Making $\A$ act on $\HH=\C^2$ as diagonal
matrices,
\begin{equation}
  \label{eq:26}
  \pi(z_1, z_2) := \left(\begin{array}{cc} z_1 &0 \\0 & z_1\end{array}\right), 
\end{equation}
with 
\begin{equation}
  \label{eq:27}
  D =\left(\begin{array}{cc} 0 &m \\ \bar m & 0\end{array}\right) \quad m\in\C
\end{equation}
as Dirac operator (the diagonal of $D$ commutes with the
representation $\pi$ and so is not relevant for the distance
computation), one easily computes that
the spectral distance between the two pure states
\begin{equation}
\delta_i(z_1, z_2)
:= z_i, \; i=1,2
\label{eq:29bbis}
\end{equation}
of $\C^2$ is
\begin{equation}
  \label{eq:28}
  d (\delta_1, \delta_2) = \frac 1{|m|}.
\end{equation}
The spectral distance thus allows to equip the
discrete two-point space $\left\{\delta_1, \delta_2\right\}$ with a
generalization of the geodesic distance, although the usual notion of length-of-the-shortest-path no longer makes sense since there is no ``points'',
i.e. no pure states, between $\delta_1$ and $\delta_2$. Incidentally,  this raises
the question of what should play the role of geodesics in noncommutative
geometry: a curve in $\sa$, in $\pa$, or something else ? We come back
to this question in \S\ref{points}  and \S\ref{sec:geodesics}.
\medskip

The construction  above generalizes to
arbitrary dimension: consider $\A=\C^N$ acting diagonally on $\C^N$, with $D$ a $N\times
N$ selfadjoint matrix with null-diagonal. For
simplicity, we restrict to Dirac operators with real entries, that is
\begin{equation}
D_{ij}=D_{ji}\in\R
\label{eq:1020}
\end{equation}
for any $i, j\in [1, N]$. One has $N$-pure states
$\delta_i, i=1,2, ..., N$ and we write the distance
\begin{equation}
d(i,j):= d(\delta_i, \delta_j).
\label{eq:96}
\end{equation}
\begin{prop}\cite[Prop. 7]{Iochum:2001fv} 
\label{distancetroispoints} For $N=3$, one deals with a three point
space with distance
\begin{equation}
d(1,2)=\sqrt{\frac{D_{13}^2+D_{23}^2}{D_{12}^{2}D_{13}^{2}+D_{12}^{2}D_{23}^{2}+D_{23}^{2}D_{13}^{2}}}.
\label{eq:3030}
\end{equation}
The other distances are obtained  by cyclic 
permutations of the indices, and verify the triangle inequality ``to the square'' 
\begin{equation}
\label{trgcarre}
d(1,2)^2 + d(2,3)^2 \geq  d(1,3)^2.
\end{equation}
\end{prop}
Formula \eqref{eq:3030} is invertible. That is, given three positive numbers  $(a,b,c)$ verifying
(\ref{trgcarre}), there exists a Dirac operator giving these numbers as
distances.
\begin{prop} \cite [Prop. 8]{Iochum:2001fv} 
Let  $a,b,c$ three positive real numbers such that
\begin{equation}
a^2+b^2\geq c^2,\quad b^2+c^2 \geq a^2, \quad
a^2+c^2\geq b^2.\label{eq:41}
\end{equation}
There exists an operator $D$ such that
\begin{equation}
d(1,2)= a,\quad d(1,3)= b,\quad d(2,3)=c.\label{eq:43}
\end{equation}
It has coefficients 
\begin{equation}
D_{12}=\sqrt{{\frac{2(b^2 + c^2 - a^2)}{(a+b+c)(-a+b+c)(a-b+c)(a+b-c)}}},\label{eq:4040}
\end{equation}
$D_{13}$ and $D_{23}$ are obtained by cyclic permutations of $a,b,c$.
\end{prop}

A surprising interpretation of \eqref{eq:3030} and \eqref{eq:4040} comes from electric
circuits \cite{Iochum:2001fv}. Starting with three numbers $a,b,c$ satisfying \eqref{eq:41},
one defines 
\begin{equation}
  \label{eq:42}
  r_1:= a^2+ b^2 - c^2,\quad r_2:= a^2 + c^2 - b^2, \quad r_3 := b^2 +
  c^2 - a^2.
\end{equation}
By \eqref{eq:43}, $d(1, 2)^2 = r_1 + r_2$ is the
resistance between the points $1, 2$ of the ``star'' circuit made of
the three resistances $r_1, r_2, r_3$ (fig. \ref{circuit}), and
similarly for $d(1,3)$ and $d(2,3)$. It is well known in electricity
that the star circuit with resistance $r_i$ is equivalent to a
triangle circuit with resistance
\begin{equation}
R_{ij}:= D_{ij}^{-2}\label{eq:12}
\end{equation}
where the $D_{ij}$'s
are precisely given by formula \eqref{eq:4040}.
So modulo the reparametrizations \eqref{eq:42} and \eqref{eq:12}, the
passage from the distances to the coefficients of the Dirac operator is
similar to the passage from the star to the triangle circuits.
\begin{figure}[h]
\label{circuit}
\begin{center}
\mbox{\rotatebox{0}{\scalebox{0.7}{\includegraphics{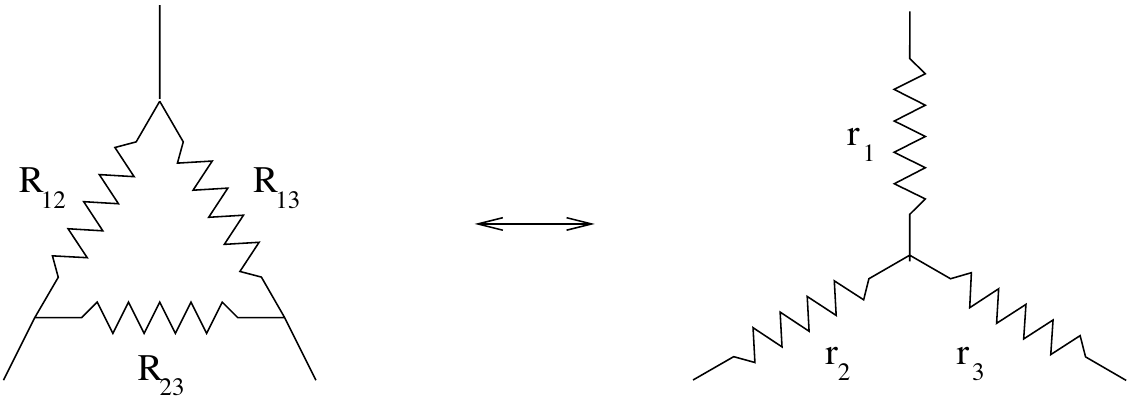}}}}
\caption{Equivalent triangle and star circuits.}\end{center}
\end{figure}

Unfortunately, the electric analogy no longer makes sense in higher
dimension. Indeed, in order to work out $\text{Lip}_D(\A)$, one needs to solve the characteristic
polynomial of the antisymmetric matrix $[D,a]$ (assuming $a$ is
real). This is of order
$\frac N2$, and is in principle not explicitly calculable as
soon as $\frac N2 \geq 5$, meaning there is little hope to explicitly compute the distance in a space with more than $N=10$
points. In fact the difficulty arises much earlier, at $N=4$. Consider
$\A=\C^4$ acting diagonally on $\HH=\C^4$ with $D$ a $4\times 4$ real
symmetric matrix with entries $D_{ij}$. Write
\begin{equation*}
\label{d4}
d_1 := \frac{1}{D_{12}},\; d_2 := {\frac{1}{ D_{13}}},\; d_3
:= { \frac{1}{ D_{14}}},\; d_4 := {\frac{1}{ D_{23}}},\; d_5 := {\frac{1}{D_{24}}},\; d_6:= {\frac{1}{D_{34}}}.
\end{equation*}
\begin{prop} \cite [Theo. 9]{Iochum:2001fv}
\label{fuck}
\begin{itemize}
\item [i.]On a four point space, $d(i,j)$ is the root of a polynomial of degree $\delta\leq
    12$, and is not in general solvable by radicals. 

\item[ii.] However there are cases where $d(1,2)$ is computable
   explicitly. For instance when
    ${\frac{1}{d_2}}={\frac{1}{d_5}}=\infty$, one has
    \begin{equation*}
      d(1,2) =\left\{\begin{array}{ll}
 d_1& \text{ when }\quad d_1^2\leq d_6^2,\\ & \\
 \frac{{d_1}\,{\sqrt{{{\left( {{{d_3}}^2} +
             {d_1}\,{d_6}\right)}^2}}}}{ {\sqrt{{{{d_1}}^2}
       {{{d_3}}^2}}}\,{\sqrt{{{{d_3}}^2} + {{{d_6}}^2}}}}  &
\text{ when } \quad d_1d_6 = d_3d_4,\\ & \\
  \sqrt{\frac{ d_1^2(d_3^2 + d_6^2)(d_4^2 + d_6 ^2)}{(d_3d_4 -
          d_1d_6)^2}}&\text{ when } \quad C \leq 0,\\
      \max\lp\sqrt{\frac{ d_1^2(d_3^2 + d_4^2)}{(d_3+d_4)^2 +
          (d_1-d_6)^2}},\; \sqrt{ \frac{ d_1^2(d_3^2 -
          d_4^2)}{(d_3-d_4)^2
          +(d_1+d_6)^2}}\rp &\text{otherwise},\end{array}\right.
\end{equation*}
      where
      \begin{equation*}
C := ({{({d_3} + {d_4}) }^2}{d_6} +
      ({d_1} - {d_6})({d_3}\,{d_4} - {{{d_6}}^2})) ({{( {d_3} - {d_4}
          ) }^2}{d_6} +({d_1}+{d_6}) \,({d_3}{d_4} + {{{d_6}}^2})).\label{eq:54}
         \end{equation*}
As well, 

\begin{equation*}
      d(1,3)   =\left\{ \begin{array}{ll} 
{\sqrt{{{{d_3}}^2} + {{{d_6}}^2}}} &\text{ when }(d_3^2 + d_6^2) \leq
(d_1d_6 - d_3d_4)^2,\\ &\\
      {\sqrt{{{{d_1}}^2} + {{{d_4}}^2}}} &\text{ when } (d_1^2
      + d_4^2) \leq (d_1d_6 - d_3d_4)^2,\\ & \\
       = \max\lp{\frac{{\sqrt{{{\left( {d_1}\,{d_3} +{d_4}\,{d_6}
                  \right)}^2}}}}{{\sqrt{{{\left( {d_3} + {d_4}\right)
                }^2} + {{\left( {d_1} - {d_6}
                  \right)}^2}}}}},{\frac{{\sqrt{{{\left( {d_1}\,{d_3}
                    +{d_4}\,{d_6} \right) }^2}}}} {{\sqrt{{{\left(
                    {d_3} -{d_4}\right) }^2} + {{( {d_1} + {d_6}
                  )^2}}}}}} \rp&\text{ otherwise}.\end{array}\right.\label{eq:37}
    \end{equation*}
  \end{itemize}
 The other distances are obtained by cyclic permutations. 
\end{prop}

This proposition shows that already for $N=4$ the distance formula cannot be inverted.
This means that  given a set of $\frac 12N(N-1)$ real positive  numbers satisfying the triangle inequality, there is no algorithm
permitting to build a $N\times N$ Dirac operator $D_N$ giving back
these numbers as the spectral distance associated to the spectral
triple $(\C^N, \C^N, D_N)$. However such an algorithm exists if one
allows the size of the Hilbert space to increase.
\begin{prop}\cite[Prop. 13]{Iochum:2001fv}
  Let $d_{ij}, 1\leq i, j\leq N, i\neq j$,  be a finite sequence of possibly infinite
  strictly positive numbers such that
  \begin{equation}
d_{ij} = d_{ji} \;\text{ and }\; d_{ij} \leq d_{ik} +
  d_{kj} \,\text{ for any }\, i, j, k.\label{eq:63}
  \end{equation}
Then there exists a spectral triple $(\A,
  \HH, D)$ with $\A = \C^N$ and $\HH = \C^{\frac 32N(N-1)}$ such that the resulting distance on
  the set of pure states of $\A$ is given by the numbers~$d_{ij}$.
\end{prop}
\noindent A similar
 construction has been proposed in \cite[\S 2.2]{Walterlivre}.
\subsection{Distances on graphs}

\label{sec:graph}
The four-point space in Prop. \ref{fuck} suggests that for $N\geq 4$,  there is little sense in 
trying to compute explicitly the distance in
a $N$-point space with the most general Dirac
operator. However some general properties of the distance can be worked
out for arbitrary $N$. To this aim we let $\mathcal{A}=\C^N$
act as diagonal matrices on $\HH= \C^N$, and we  identify the $N$ pure states of $\C^N$,
\begin{equation}
\delta_i(z_1, z_2, ..., z_N) := z_i\quad\quad \forall (z_1, z_2, ... , z_N)\in\C^N,
\label{eq:380bis0}
\end{equation}
with the points $1, 2, ..., N$ of a $N$-point graph. We take as a
Dirac operator the incidence matrix of the graph, that is 
\begin{equation}
D= \left(
\begin{array}{ccccc}
0    &D_{12}&\ldots&\ldots&D_{1N} \\
D_{12}&      &D_{23}&      &0 \\
\vdots&D_{23}&0    &\ddots&\vdots   \\
\vdots&     &\ddots&\ddots& D_{N-1,N}   \\
D_{1N}&\ldots&\ldots &D_{N-1,N}& 0 \\
\end{array}
\right) \;\quad\; D_{ij} \in \mathbb{R},\label{eq:31}
\end{equation}
where $D_{ij}=D_{ji}$ is non-zero if and only if there is a link in
the graph between the points $i$ and $j$. A path $\gamma_{ij}$ is a sequence of $p$ distinct points
$(i,i_2,...,i_{p-1},j)$ such that 
$$D_{i_{k}i_{k+1}}\neq 0\; \text{ for all } k\in\{1,p-1\}.$$
Two points $i$,$j$ are said connected if there exists at least one
path $\gamma_{ij}$.  We define the length of a path $\gamma_{ij}$ as
\begin{equation*}
L(\gamma_{ij}):= \underset
{k=1}{\overset{p-1}\Sigma}{\frac{1}{|D_{i_ki_{k+1}}|}},
\end{equation*}
and the geodesic distance $L_{ij}$ between any two connected points $i, j$
as the length of the
shortest path $\gamma_{ij}$.

\begin{prop}\cite [Prop. 4]{Iochum:2001fv}
  \label{chemin}
  \begin{itemize}
\item[i.] 
Let $D'$ be the operator obtained by canceling one or more
    lines and the corresponding columns in $D$, and $d'$ the
    associated distance. Then $d' \geq d$.

 \item[ii.]
The distance between two points $i$ and $j$
    depends only on the matrix elements corresponding to points
    located on paths $\gamma_{ij}$.

 \item[iii.]
 The distance between any two points is
    finite if and only if they are connected.

\item[iv.] For any two points $i,j$, one has $d(i, j) \leq L_{ij}$. 
 \end{itemize}
\end{prop}
\medskip

A case explicitly computable is the maximally connected graph, that is
the operator $D$ with all
coefficients equal to a fixed real constant $k$.
\begin{prop}\cite[Prop. 5]{Iochum:2001fv} 
  \begin{itemize}
  \item[i.] The distance between any two points $i, j$ is
\begin{equation}
d(i,j) ={\frac{1}{|k|}}{\sqrt{\frac{2}{N}}}.
\label{eq:35bis}
\end{equation}

\item[ii.] If the link between two points $i_1,i_2$ \-- and only this
  link \-- is cut, $D_{i_1i_2}=0$, then
  \begin{equation}
 d(i_1,i_2) ={\frac{1}{|k|}}{\sqrt{\frac{2}{N-2}}}.
\label{eq:36}
\end{equation}
\end{itemize}
\end{prop}

Examples of explicit computation of the spectral distance in lattices
can be found in \cite{bimonte}, \cite{pekin},
\cite{DAndrea:2013kx}. Applications to quantum gravity have been
explored in~\cite{CR2014}.

\subsection{Projective spaces}
\label{sec:proj}
The space of pure states of  $M_n(\cc)$, $n\in\N$, is the  projective space
$\C P^{n-1}$: any normalized vector $\xi\in\C^n$ defines
the pure state 
\begin{equation}
\label{notationcp}
\omega_\xi (a)
= \scl{\xi}{a\xi} \quad \forall a\in M_n(\C)
\end{equation}
where $\scl{\cdot}{\cdot}$ is the usual  inner product on $\C^n$. Two such vectors equal up to a phase define the same state, and any
pure state comes in this way. All the  representations of $M_n(\C)$
induced by these pure states via the Gelfand-Neimark-Segal
construction are  equivalent,  that is why it is sometimes argued
\cite{madore} that $M_n(\cc)$
should be considered as a $1$-point space. On the contrary, we argue
 that the spectral distance  provides the space of pure states of $\mn$
with structure finer than the one of irreducible representations,
and there is no reason to neglect it. We come back to this point in
section \ref{points}.

We consider the spectral triple
\begin{equation}
\A=M_n(\C), \quad  \HH= \C^n, \quad D
\label{eq:146}
\end{equation}
where the action of $\A$ on $\HH$ is the usual representation of
matrices, while $D$  is an arbitrary selfadjoint element of $M_n(\C)$.  There exists no explicit computation of the distance in the
most general
case, that is between any two states of  $M_n(\C)$ for arbitrary
$n$. There are such computations for $n=2$, which are  the object of
\S \ref{section:ball}. For $n\geq 2$, we expose below some properties
of the connected components, which are
are slight generalizations of unpublished results 
of \cite{Martinetti:2001fk}. There is also an
explicit computation of the distance between any  pure states of $M_n(\C)\oplus \C$ but
with a particular class of operator $D$, presented in
Prop. \ref{sphere+point}.

From now on we assume $n\geq 2$. To make the correspondence between $\C P^{n-1}$
and normalized vectors in $\C^n$ explicit, it is convenient to fix as a basis
of $\HH$ an orthonormal set of eigenvectors $\psi_i$ of $D$, $i=1,n$,
so that
\begin{equation}
D =\text{diag }(d_1, d_2, ..., d_n)
\label{eq:39}
\end{equation}
where $ d_i \in \R$ are the eigenvalues of $D$ (possibly null). For
any eigenvector $\psi_i$, we call \emph{eigenstate of $D$ } the pure state
\begin{equation}
\omega_i := \scl{\psi_i}{\cdot\; \psi_i}.
\label{eq:44}
\end{equation}
We write $e_{ii}$ the diagonal
matrix with only non-zero component 
the $i^\text{th}$ entry that is equal to $1$. Given a normalized complex $n$-vector $\xi$, we
write $\xi_i := \scl{\xi}{\psi_i}$ its components on the eigenbasis of
$D$, and 
\begin{equation}
  \label{eq:45}
  \ox := \scl{\xi}{\cdot\; \xi}
\end{equation}
the corresponding pure state of $M_n(\C)$. It is not difficult to
characterize the pure states at finite distance from one another.
\begin{prop}
\label{pointsinfinis} 
Let $\xi, \zeta$ be 
normalized vectors in $\C^n$. The distance between the pure states $\ox$ and $\oz$ is finite if and only if the
projections of $\xi$ and $\zeta$ on the kernel - as well as on any
eigenspace of $D$ - are equal up
to a phase. That is, for any eigenspace $\HH_J$ of $D$ ($J\geq 1$ an integer),
there exists a phase $\theta_J\in
[0, 2\pi[$ such that 
\begin{equation}
  \label{eq:17}
  \xi_i = e^{i\theta_J} \zeta_i \quad\text{ for any } i\in I_J,
\end{equation}
where $I_J$ is the subset of $\left\{1, n\right\}$ such that $\HH_J =
\text{span} \left\{e_{ii}, i\in I_J\right\}$. 
\end{prop}
\begin{proof}
By lemma \ref{infinitelem}, we just need to show that $\ox, \oz$
coincide on the kernel of $L_D$ if and only if \eqref{eq:17} holds true. Let us first assume there is no degeneracy,
that is all the eigenvalues $d_i$ of $D$ are distinct. Then \eqref{eq:17} amounts to
 \begin{equation}
\label{eqq:0}
   |\xi_i| = |\zeta_i| \quad \forall i=1, ...,n.
 \end{equation} 
The kernel of $L_D$ is the set of diagonal matrices. Any two states $\ox,\oz$ coincide 
on the kernel if and only if  $\ox(e_{jj}) = |\xi_j|^2$  equals $\oz(e_{jj}) = |\zeta_j|^2$ 
for any $j$, that is equation \eqref{eqq:0}. 

In case of degeneracy, one has 
\begin{equation}
  \label{eq:59}
  \text{Ker}\,L_D = \underset{J\in\N}{\bigoplus} {\mathcal B}(\HH_J)
\end{equation}
 and $\varphi, \varphi'$
coincide on each summand if and only if \eqref{eq:17} holds.
\end{proof}
As a corollary, one obtains that the connected component in the space
of pure states of any pure state is a torus inside $\C P^{n-1}$.
\begin{cor}
  Let $\ox$ be the pure state of $M_n(\C)$ associated to the  normalized vector $\xi\in\C^n$ with components $\xi_i$ in the
  eigenbasis of $D$. Then
$\text{Con}(\ox)\cap {\mathcal P}(M_n(\C))$
is  
  the $k-1$-torus 
  \begin{equation}
    \label{eq:68}
   U_\xi =\left\{\left(\begin{array}{c} \xi_i \quad \forall
       i\in I_1\\ \xi_i e^{i\theta_2}\quad \forall
       i\in I_2\\ \vdots \\ \xi_i e^{i\theta_k}\quad \forall
       i\in I_k\end{array}\right), \theta_2, ..., \theta_k
   \in [0,2\pi[\right\}, 
  \end{equation}
where $k$ is the number of distinct eigenvalues of D.
\end{cor}
\begin{proof}
  This follows directly from Prop. \ref{pointsinfinis}. Notice that if
  $k=1$, that is $D$ is proportional to the identity, then
  \eqref{eq:17} indicates that $\oz$ is at finite distance from $\ox$
  if and only if $\xi = e^{i\theta_1}\zeta$. But this means 
  $\ox=\oz$, so that $\text{Con} (\ox)$ reduces to $\ox$ itself. That
  is why in \eqref{eq:68} the phase $e^{i\theta_1}$ is factorized out. 
\end{proof}
\medskip

\subsection{The three dimensional ball}
\label{section:ball}

For $n=2$, the distance can be explicitly computed. The space of pure
states of $M_2(\C)$ is the complex projective plane $\mathbb{C}P^{1}$, which is  in $1$-to-$1$
correspondence with the $2$-sphere: to any  normalized complex
vector $\xi\in\C^2$ with components $\xi_1,  \xi_2$,
one associates the point $p_\xi$ of $S^2$ with Euclidean coordinates
\begin{equation}
\label{hopf}
x_{\xi} := 2\re(\xi_1\bar{\xi}_2), \quad y_{\xi} := 2\im(\xi_1
\bar{\xi}_2)\,  \text{ and } \, z_{\xi} :=  |\xi_1|^2 -|\xi_2|^2.
\end{equation}
The evaluation of $\ox$ on $a\in M_2(\C)$ with components $a_{ij}$
reads (see e.g. \cite[\S 4.3]{Cagnache:2009oe})
\begin{equation*}
  \ox(a) = \frac{1+z_\xi }2\, a_{11} + \frac{1- z_\xi }2\, a_{22} + r\,\Re \left(e^{i\Xi} a_{12}\right)
\end{equation*}
 where
 \begin{equation}
r e^{i\Xi} := x_\xi + i y_\xi = 2\xi_1\bar \xi_2.
\label{eq:72}
\end{equation}

A non-pure state $\varphi$ of $M_2(\C)$ is given by a probability distribution $\phi$ on $S^2$,
\begin{equation}
  \label{eq:38}
  \varphi (a) = \int_{S^2} \phi(\xi) \,\omega_\xi(a)\, d\xi
 \;=\; \frac{1+\tilde z_\phi }2\, a_{11} + \frac{1- \tilde z_\phi }2 \,a_{22} + \tilde r_\phi\,\Re \left(e^{i\tilde\Xi_\phi} a_{12}\right)
\end{equation}
where 
$d\xi$ is the $SU(2)$ invariant measure on $S^2$ normalized to $1$ and
\begin{equation}
{\bf \tilde x}_\phi:= (\tilde x_\phi, \, \tilde y_\phi,\,  \tilde z_\phi)
\label{eq:73}
\end{equation}
denotes the mean
point of $\phi$, that is
\begin{equation*}
  \tilde x_\phi  := \int_{S^2} \phi(\xi)\, x_\xi \,d\xi
\end{equation*}
with similar notation for $\tilde y_\phi$, $\tilde z_\phi$, and
$\tilde r_\phi e^{i\tilde\Xi_\phi}:= \tilde x_\phi + i\tilde y_\phi$. The
correspondence between a state and a mean point,
\begin{equation}
  \label{eq:74}
  \varphi \longleftrightarrow {\bf \tilde x}_\phi,
\end{equation}
 is $1$-to-$1$, that is ${\mathcal
  S}(M_2(\C))$ is the 3-ball. But unlike the commutative case,  two
distinct probability measures
may have the same mean point, so that the correspondence between
states and probability measures is not $1$-to-$1$.
\smallskip

Let us first consider the spectral triple \eqref{eq:146}, that is
$M_2(\C)$ acting on $\C^2$ with $D$ an arbitrary selfadjoint $2\times
2$ matrix. As in \S \ref{sec:proj}, we chose as basis of
$\C^2$ an orthonormal eigenbasis of $D$, so that the north and south poles of $S^2$ are the image of the eigenstates
of $D$ by \eqref{hopf}.
\begin{prop}
 \label{propmercrediun}
Assume the two eigenvalues $d_1$, $d_2$ of $D$
are distinct (otherwise $D$ is proportional to the identity and all
the distances are infinite).
 \begin{itemize}
 \item[i.] Two points ${\bf \tilde x}_\phi, {\bf \tilde x}_{\phi'}$ in
   the 3-ball are at finite distance iff $\tilde z_\phi = \tilde z_{\phi'}$. 
\item[ii.] The distance between two points ${\bf \tilde x}_\phi, {\bf
    \tilde x}_{\phi'}$ with the same $z$-coordinate is proportional to
  the
  chord distance on the circle:
  \begin{equation}
d({\bf \tilde x}_\phi, {\bf \tilde
  x}_{\phi'})=\frac{1}{|d_1-d_2|}\sqrt{(\tilde x_{\phi} - \tilde x_{\phi'})^2+(y_{\phi}-y_{\phi'})^2}.
\label{eq:32}
\end{equation}
\end{itemize}
\end{prop}
\begin{proof}
  The result has been shown for pure states in
  \cite[Prop. 2]{Iochum:2001fv}{\footnote{Notice the misprint of a factor $2$  in the result as
  expressed below Prop. 2}}. The proof easily adapts to non-pure states as follows. If
  $\tilde z_\phi \neq \tilde z_{\phi'}$, then for
  \begin{equation}
      b=\left(\begin{array}{cc}
          1&0\\ 0&0 \end{array}\right)\in\text{Ker} \, L_D
        \end{equation}
one has by \eqref{eq:38} that $\varphi(b) \neq \varphi'(b)$, meaning
the distance is infinite by lemma \ref{infinitelemfinitedim}.

Assume   $\tilde z_\phi = \tilde z_{\phi'}$. Then 
\begin{align}
|\varphi(a) -\varphi'(a)| &= \Re \left( a_{12} \left(\tilde x_{\varphi} - \tilde x_{\varphi'} + i(\tilde y_{\varphi} - \tilde y_{\varphi'})\right)\right)\\
&\leq |a_{12}||\left(\tilde x_{\varphi} -  \tilde x_{\varphi'} + i(\tilde y_{\varphi} - \tilde y_{\varphi'})\right)|.
  \end{align}
A  direct calculation shows that 
  \begin{equation}
    L_D(a) = |a_{12}| |d_1 - d_2| 
  \end{equation}
so that 
\begin{equation}
    d(\varphi, \varphi') \leq \frac 1{|d_1 - d_2|}\sqrt{
\left(\tilde x_{\varphi} - \tilde x_{\varphi'}\right)^2 + \left(\tilde y_{\varphi} -
      \tilde y_{\varphi'}\right)^2}.
  \end{equation}
This upper bound is attained by $a=a^* \in M_2(\C)$ with components
$a_{11}=a_{22}=0$ and $a_{12}= \frac 1{|d_1 - d_2|}e^{-i\theta}$ with
$\theta = \text{arg} \left(\tilde x_{\varphi} - \tilde x_{\varphi'} + i(\tilde y_{\varphi} - \tilde y_{\varphi'})\right)$.
\end{proof}
The proposition above shows that the simplest spectral triple on $M_2(\C)$ equips the $3$-ball with a
metric that slices the ball into circles at infinite distance from one
another, in particular the poles of $S^2$ are at infinite distance from any
   other state. To avoid such infinities,  according to lemma
\ref{infinitelem} one needs to reduce the kernel of the semi-norm $L_D$ to
the multiples of the identity. This can be done by changing the space
of representation and the operator $D$. An exemple is the following
spectral triple, which comes from the truncation of the spectral
triple of the Moyal plane in~\S ~\ref{section:Moyal}. Namely, one makes
\begin{equation}
\A= M_2(\C)
\;\text{ act on }\; \HH = M_2(\C)\otimes \C^2 \; \text{ as } m\otimes \I_2,
\label{eq:147}
\end{equation}
where the
element $m$ of the algebra $M_2(\C)$ acts on the Hilbert space $M_2(\C)$ by matrix multiplication. As a Dirac
operator, one takes
\begin{equation}
  \label{eq:61}
  D= -i\sqrt 2\left(\begin{array}{cc} 0_2 & [X^\dagger, \cdot] \\ -[X,
      \cdot ] & 0_2\end{array}\right)
\end{equation}
where the non-zero terms are the commutators with the matrix
\begin{equation}
  \label{eq:71}
  X= \frac 1{\sqrt \theta} \left(\begin{array}{cc} 0&0 \\ 1&0\end{array}\right)
\end{equation}
and its adjoint. This operator is the restriction to $M_2(\C)$ of the
usual Dirac operator of the plane acting on $L^2(\R^2)$ (see \cite{Cagnache:2009oe}
for details).

\begin{prop}\cite[Prop. 4.4]{Cagnache:2009oe}.
\label{propdista2} The spectral distance between any two states of
$M_2(\C)$, identified to points ${\bf \tilde x}_\phi, {\bf \tilde
  x}_{\phi'}$ of the 3-ball by \eqref{eq:74} is finite. More exactly, 
\begin{equation*}
d({\bf \tilde x}_\phi, {\bf \tilde
  x}_{\phi'})=
\sqrt{\frac{\theta}{2}}\times\begin{cases}
d_{eq} ({\bf \tilde  x_\phi}, {\bf \tilde  x_{\phi'}}) & \mathrm{if}\;\;|\tilde z_\phi-\tilde z_{\phi'}|\leq d_{eq} ({\bf \tilde  x_\phi}, {\bf \tilde  x_{\phi'}}) \;,\\
\frac{d_{Ec} ({\bf \tilde  x_\phi}, {\bf \tilde  x_{\phi'}})^2}{2 |\tilde z_\phi-\tilde z_{\phi'}|}  & \mathrm{if}\;\;|\tilde z_\phi-\tilde z_{\phi'}|\geq d_{eq} ({\bf \tilde  x_\phi}, {\bf \tilde  x_{\phi'}}),  \;
\end{cases}
\end{equation*}
 where
 \begin{equation}
 d_{Ec} ({\bf \tilde  x_\phi}, {\bf \tilde  x_{\phi'}}) =  \sqrt{\abs{\tilde x_\phi - \tilde x_{\phi'}}^2 + \abs{\tilde y_\phi - \tilde y_{\phi'}}^2+ \abs{\tilde z_\phi - \tilde z_{\phi'}}^2 }
\label{eq:75}
\end{equation}
denotes the euclidean distance on $B^3$ while
\begin{equation}
 d_{eq} ({\bf \tilde  x_\phi}, {\bf \tilde  x_{\phi'}})=
 \sqrt{\abs{\tilde x_\phi - \tilde x_{\phi'}}^2 + \abs{\tilde y_\phi -
     \tilde y_{\phi'}}^2 }
\label{eq:76bis}
 \end{equation}
is the Euclidean distance between the projections of the points on the
equatorial plane $z=0$.
\end{prop}

 \noindent Contrary to the simplest spectral triple on $M_2(\C)$ of
proposition \ref{propmercrediun}, with the spectral triple
\eqref{eq:147} inherited
from the Moyal plane, the spectral distance induces on $B^3$ the Euclidean
topology, which coincides with the weak$^*$ topology \cite[\S
4.3]{Cagnache:2009oe}.
 \medskip


For sake of completeness, let us mention another example of finite
dimensional spectral triple that allows to
orientate the $3$-ball, by adding one point at finite distance from one
of the
pole of $S^2$. This is obtained by letting
\begin{equation}
\mathcal{A}=M_{n}(\mathbb{C})\oplus\mathbb{C}\quad \text{ act on }\quad
\mathcal{H}=\mathbb{C}^{n}\oplus\mathbb{C}
\label{eq:149}
\end{equation}
as
\begin{equation}
\label{rep2points}
a=\left(\begin{array}{cc} x & 0\\ 0 &y\end{array}\right),
\end{equation}
with $x \in M_n(\mathbb{C})$ and $y\in \mathbb{C}$. 
As a Dirac operator, one takes
\begin{equation}
\label{dfi}
D = \left(\begin{array}{cc} 0_n & v \\ v^* & 0_n \end{array}\right).
\end{equation}
where $v\in\C^n$. 
  \begin{prop}\cite [Prop. 3]{Iochum:2001fv}
\label{sphere+point}
For two pure states $\ox$, $\oz$ of $M_n(\C)$ such that $\xi_j=
e^{i\theta}\zeta_j$ for all $j\in[2,n]$, the distance is 
\begin{equation}
d(\ox, \oz) = \frac{2}{\norm{v}} \sqrt{1 - \abs{\scl{\xi}{\zeta}}^2}.
\label{eq:77}
\end{equation}
Furthermore, the pure state $\omega_c$ of $\C$ is at infinite distance
from all the pure states of $M_n(\C)$, except 
$\omega_{v}$ for which
$$d(\oc, \omega_{v}) = \frac{1}{\norm{v}}. $$
\end{prop}

Applied to $M_2(\cc)\oplus \cc$, one has that the space of pure states
is the disjoint union of the sphere $S^2$ and the point $\oc$. On the
sphere the condition of finitude of the distance is the
same as in proposition \ref{propmercrediun}: $S^2$ is sliced in
circles at infinite distance from one another, while on each circle the distance is proportional to the Euclidean distance
on the disk. The
pure state $\omega_{v}$ gives the north pole of the sphere, and is at finite distance from
$\oc$. In this sense  adding a point allows to give an orientation to the
sphere, by distinguishing between the south pole at infinite distance
from any other points and the north pole at finite distance from the
isolated point $\omega_c$.


 \section{Product of geometries and the Higgs field}
\label{section:almostcg}
We now consider the metric aspect of the product of
spectral triples. Recall that a spectral triple $(\A, \HH, D)$ (with
representation $\pi$) is
\emph{graded} if there exists a grading $\Gamma$ of $\HH$ (that is
a selfadjoint operator $\Gamma$  such that $\Gamma^2=\I$) which satisfies
\begin{equation}
  \label{eq:150}
\Gamma D = - D \Gamma, \quad   [\Gamma, \pi(a)]=0\quad \forall a\in\A.
\end{equation}
Given two spectral triples $T_1 = (\A_1, \HH_1, D_1)$,
$T_2 = (\A_2, \HH_2, D_2)$ where we assume that $T_1$ is graded
with grading $\Gamma_1$, the product
\begin{equation}
  \label{eq:78}
  \A=\A_1 \otimes \A_2,\quad \HH=\HH_1\otimes \HH_2,\quad D= D_1\otimes \I_2
  + \Gamma_1\otimes D_2
\end{equation}
is again a spectral triple \cite{Connes:1996fu} that we denote
\begin{equation}
  \label{eq:83bis}
  T:= T_1 \times T_2.
\end{equation}
We are interested in
the spectral distance $d$ associated to $T$, and how it is
related to the distance $d_1, d_2$ associated to $T_1$ and $T_2$.

General  results on that matter are recalled in \S
\ref{section:pythagoras}. In \S \ref{section:almost} we focus on
the case where $T_1$ is the usual spectral triple of a closed (spin) manifold 
described in  \eqref{eq:13} - \eqref{eq:142} and $T_2$ a finite dimensional spectral
triple as investigated in section \ref{section:finite}. In that case,
the product $T$ describes a slightly noncommutative generalization of a
manifold, called \emph{almost commutative geometry}, which is important
for physical applications since it is at the hearth of the
description of the standard model of particle physics, as explained in
\S \ref{higgs}. The bundle structure of the space of pure states then
also opens interesting links with
sub-Riemannian geometry. This is the object of section \ref{subriemannian}.

\subsection{Pythagoras inequality}
\label{section:pythagoras}

Till recently, the metric aspect of product of spectral triples had
been studied mainly for almost commutative geometries. In particular, it came out
that for the spectral triple describing the standard model of
elementary particles, the
distance $d$ between pure states satisfies the Pythagoras equality with respect to the
distances $d_1$ on the manifold  and the distance $d_2$ of the finite dimensional
spectral triple describing the gauge degrees of freedom \cite{Martinetti:2002ij}.  
A similar result was found for the product of the Moyal plane with the
two-point space of  \S \ref{sec:discretespaces} \cite{Martinetti:2011fko}. 
This raises the
question whether the product \eqref{eq:78} is always orthogonal in the
sense of Pythagoras theorem. By this we intend that given two
separable{\footnote{A state $\varphi\in\sa$ is said separable if it decomposes as the
product of two states $\varphi_1\in{\mathcal S}(\A_1)$,
$\varphi_2\in{\mathcal S}(\A_2)$}} states~of~$\A=\A_1\otimes \A_2$,
\begin{equation}
  \label{eq:84}
  \varphi:=  \varphi_1 \otimes \varphi_2,\quad \varphi':=\varphi'_1\otimes \varphi'_2,
\end{equation}
does one have - at least between pure states - that 
\begin{equation}
d^2(\varphi, \varphi')  \quad\text{equals}\quad d_1^2\,(\varphi_1, \varphi'_1) +
d_2^2\,(\varphi_2, \varphi'_2)\;?
\label{eq:85bis}
\end{equation}

In
\cite{DAndrea:2012fkpm} we proved the following Pythagoras inequalities for
the product of arbitrary unital spectral triples. For a complete and
pedagogical treatment on that matter, as well as some significant generalizations, we invite the reader to see the contribution
of F. D'Andrea in this volume \cite{DAndrea:2015ab}.
\begin{thm}\cite[Theo. 5]{DAndrea:2012fkpm}
\label{thm:Pythagoras}
Given the product \eqref{eq:78} of two spectral triples $(\A_i, \HH_i,
D_i)$, $i=1,2$ and two separable states
$\varphi=\varphi_1\otimes\varphi_2$ and
$\varphi'=\varphi_1'\otimes\varphi_2'$ of $\A$, one has:
\begin{equation}
\label{ineqPyth:1}
d(\varphi,\varphi') \leq d_1(\varphi_1,\varphi'_1)+d_2(\varphi_2,\varphi'_2) \;.
\end{equation}
Furthermore, if the spectral triples are unital, then
\begin{equation}
d(\varphi,\varphi') \geq \sqrt{d_1(\varphi_1,\varphi'_1)^2+d_2(\varphi_2,\varphi'_2)^2} \;.
\label{eq:mainB}
\end{equation}
\end{thm}
Combining \eqref{eq:mainB} and \eqref{ineqPyth:1} one obtains a
noncommutative version of Pythagoras theorem that
holds true for any separable states in the product of arbitrary unital
spectral triples (it was first proven in \cite[Prop.~II.4]{Martinetti:2011fko} for
pure states, with 
one of the spectral triples the two-point space $\C^2$).
\begin{cor}\cite{DAndrea:2012fkpm} Let
  $\varphi=\varphi_1\otimes\varphi_2$,
  $\varphi'=\varphi'_1\otimes\varphi'_2$ be two separable states in
  the product of two unitary spectral triples. Then
  \begin{equation}
   \sqrt{d_1(\varphi_1,\varphi'_1)^2+d_2(\varphi_2,\varphi'_2)^2}\leq
  d(\varphi,\varphi') \leq
  \sqrt{2}\sqrt{d_1(\varphi_1,\varphi'_1)^2+d_2(\varphi_2,\varphi'_2)^2}
  \;.\label{eq:79}
  \end{equation}
Furthermore these inequalities are optimal, in that there exist
examples that saturate the bounds.
\end{cor}
Notice that \eqref{ineqPyth:1} is not the triangle inequality
\begin{equation}
  \label{eq:80bisb}
  d(\varphi, \varphi') \leq d(\varphi_1\otimes \varphi_2, \varphi_1\otimes\varphi'_2) +
  d(\varphi_1\otimes \varphi'_2, \varphi'_1\otimes\varphi'_2) 
\end{equation}
because
nothing guarantees that the distance between two states $\varphi_1\otimes\varphi_2$,
$\varphi_1\otimes \varphi'_2$ that differ only on one of the algebras gives back the distance
on a single spectral triple, that is $d(\varphi_1\otimes \varphi_2, \varphi_1\otimes\varphi'_2)$
    equals $d_2(\varphi_2, \varphi'_2)$. 
    In fact, this comes out as a corollary, initially
proven in \cite{Martinetti:2002ij}. 
\begin{cor}\cite[Cor. 6]{DAndrea:2012fkpm} 
\label{coroll:distancereduite} Let
  $\varphi=\varphi_1\otimes\varphi_2$,
  $\varphi'=\varphi'_1\otimes\varphi'_2$ be two separable states in
  the product of two unitary spectral triples. If
  $\varphi_2=\varphi'_2$, then $d(\varphi,\varphi')=d_1(\varphi_1,\varphi'_1)$,
and similarly if $\varphi_1=\varphi'_1$ then
$d(\varphi,\varphi')=d_2(\varphi_2,\varphi'_2)$.
\end{cor}
\medskip

To conclude the generalities on the product of spectral triples, let us mention an application of the projection lemma \ref{projectionlem}. It is not of
great interest in se because of the strong conditions required, but it 
turns out to be extremely useful to compute the distance in the
standard model of elementary particles, as explained in the next
subsection. Let us consider the product \eqref{eq:78} and
restrict the attention to normal states for, say, the algebra
$\A_2$. To any such state $\varphi$ is associated a support, namely a projection $s\in\aa_2$ such that $\varphi$ is faithful on $s\aa_2 s$. For a pure state $\omega$, being normal implies 
\begin{equation}
  \label{3eq:3}
  sas = \omega(a) s\quad \forall a\in\aa_2.
\end{equation}
We say that two normal pure states $\omega_1, \omega_2$ are in direct sum if
\begin{equation}
s_1 a s_2 = 0 \quad \forall a\in\aa_2.\label{eq:87}
\end{equation}
If furthermore the sum $p=: s_1 + s_2$ of their support
commutes with $D_2$, then the distance in the product \eqref{eq:78} projects down to a two point-case $\aa_1\ot \cc^2$. 
 \begin{prop}\cite{Martinetti:2002ij}
\label{reduction}  Let $d$ be the distance associated with
the product $T= T_1\times T_2$. Let $\omega_2$, $\omega_2'$ be normal pure states of
$\aa_2$ in direct sum, and whose sum of supports $p$ commutes with
$D_2$. Then for any pure states $\omega_1$, $\omega'_1$ of $\A_1$ one has 
\begin{equation}
d\lp \omega_1\otimes \omega_2, \omega_1'\otimes \omega_2'\rp =
d_e\!\lp\omega_1\otimes\omega_c, \omega_1'\otimes \omega'_c\rp
\label{eq:82}
\end{equation}
where $\omega_c,\omega'_c$ are the two pure states of $\,\cc^2$ while $d_e$ is the distance
associated
to the product  $T_e:= T_1\times T_r$ where $T_r :=(\A_r, \HH_r,
D_r)$ with
\begin{equation}
 \aa_r:=\cc^2,\quad \hh_r:=
 p \hh_2,\quad
D_r:= p D_2 p  \big|_{\hh_r}.
\label{eq:88}
\end{equation}
\end{prop}
Note that this proposition remains true for an algebra $\aa_2$ on a
field other than $\cc$, assuming that the  notion of states is
still available. For instance in the standard model one deals with real algebras.

\subsection{Almost commutative geometries and fluctuation of the metric}
\label{section:almost}
 
A slightly noncommutative generalization of a manifold
is obtained by taking the product \eqref{eq:78} of the 
spectral triple of a closed, spin manifold $\M$, that is (see \eqref{eq:Diracoperateur})
\begin{equation}
T_1 = (\cinf, L^2(\M, S), \ds),
\label{eq:23}
\end{equation}
by a finite
dimensional spectral triple $T_2 = (\A_F, \HH_F, D_F)$. Namely one considers
\begin{equation}
  \label{eq:88bis}
  \A = \cinf \otimes \A_F, \; \HH= L^2(M, S)\otimes \HH_F,\;
  D=\ds\otimes \I_F + \gamma^5\otimes D_F
\end{equation}
where $\I_F$ is the identity operator on $\HH_F$ and $\gamma^5$ is the
grading of $L^2(\M,S)$ given by the product of the Dirac
matrices. The center of $\A$ is infinite dimensional (as an algebra)
while the noncommutative part is finite dimensional, hence the name
\emph{almost commutative geometries} often used to describe \eqref{eq:88bis}.
\smallskip

Because $\cinf$ is nuclear, the space of pure states $\pa$ of $\A$ 
 is \cite{Kadison1983}
 \begin{equation}
 \pp(\A) \simeq \pp(\cinf) \times \pp(\aa_F),
\label{eq:86}
 \end{equation}
and $\sa$ its convex hull.  $\pp(\cinf)$ is homeomorphic to $\M$ and $\pp(\aa_F)$ carries a
 natural action of the special unitarie group $SU(\A_F)$ of $\A_F$,
 \begin{equation}
   \label{eq:16}
   \omega\to \omega\circ \alpha_u \quad \forall \omega \in \pp(\A_F)
 \end{equation}
with $\alpha_u$ the inner automorphism of $\A_F$ given by conjugate
action of $u\in  SU(\A_F)$. In other terms,  $\pp(\A)$ is a trivial
$SU(\A_F)$-bundle on $\M$ with fiber $\pp(\A_F)$.
 \medskip

In the study of noncommutative algebras (or more generally
noncommutative rings), there exists a notion of equivalence
which is weaker than isomorphism but turns out to be very fruitful,
that of Morita equivalence. Given a spectral triple $(\A, \HH, D)$,
there is a generic procedure to export the
geometrical structure to a Morita equi\-valent algebra
\cite{Connes:1996fu}. Taking advantage of the self-Morita equivalence of $\A$, this
procedure yields a natural way to introduce a connection in the geometry $(\A, \HH,
D)$, by substituting the operator $D$ with a \emph{covariant Dirac
  operator} $D_A$, such that $(\A, \HH, D_A)$ is still a spectral
triple. Explicitly, this covariant operator is
\begin{equation}
\label{diraccov} D_A := D + A + JAJ^{-1},
\end{equation} where $A$ is a selfadjoint element of the set of
generalized $1$-forms \footnote{We use Einstein summation over repeated indices in alternate
positions (up/down).}
\begin{equation}
  \Omega^1_D:=\left\{ a^i [D, b_i], \quad a^i, b_i \in \A\right\},
\end{equation}
and $J$ is the \emph{real structure}. The latter is a generalization to the
non-commutative setting of the charge conjugation operator on
spinors. The only thing we need to know about it at the moment is that for any $a, b\in \A$ one has  $[JaJ^{-1},b]=0$,
so that substituting $D$ with $D_A$ in the distance formula yields
\begin{equation}
\label{distancefluc} d_A(\varphi, \varphi') := \suup{a\in\aa}\left\{
\abs{\varphi(a) - \varphi'(a)}\,,\; \norm{[D +A, a]}\leq 1\right\}.
\end{equation}
There is no reason  for $||\left[D_A,a\right]||$ to equal $||\left[D,a\right]||$,
neither for the distance $d_A$ computed with the covariant Dirac operator $D_A$ to equal the distance $d$ computed with the initial operator
$D$. That is why one talks of a
{\emph{fluctuation of the metric}}.

For almost commutative geometries \eqref{eq:88bis},   a generalized
$1$-form in $\Omega^1_D$ is \cite{D.-Kastler:1993aa}:
 \begin{equation}
A = -i\gamma^\mu f_\mu^i \ot m_i + \gamma^5 h^j\ot n_j
\label{eq:9k}
\end{equation}
where
$m_i\in\aa_F$, $h^j, f_\mu^i\in\cinf$, while  
\begin{equation}
n_j\in\Omega_{D_F}^1:= \left\{ a^i[D_F, b_i], \; a^i, b_i \in
  \A_F\right\}.
\label{eq:81bis}
\end{equation}
Omitting the tensor product,  a selfadjoint $1$-forms $A=A^*\in\Omega^1_D$ thus decomposes as the
sum
\begin{equation}
  \label{eq:151}
 A= -i\gamma^\mu A_\mu + \gamma^5 H,
\end{equation}
where 
\begin{equation}
\label{gaugg}
A_\mu:= f^i_\mu m_i
\end{equation}
is a  $\A_F$-valued skew-adjoint
1-form field over $\mm$, and 
\begin{equation}
\label{gaugh}
H:= h^j
n_j
\end{equation}
is a $\Omega^1_{D_F}$-valued selfadjoint scalar field.
 The part
of the covariant Dirac operator $D_A$ relevant for the distance
formula is the fluctuated operator
\begin{equation}
\label{drelevant}
\ds + A = \ds + \gamma^5
 H -i\gamma^\mu A_\mu  . 
\end{equation}

We investigate below how the distance on the bundle of pure
states \eqref{eq:86} is affected by the two pieces of the fluctuation: the scalar fluctuation
$H$ in \S \ref{higgs}, and the gauge fluctuation $A_\mu$ in \S \ref{subriemannian}.

\subsection{Two sheet models and the metric interpretation of the
  Higgs}\label{higgs}

Let us consider a scalar fluctuation of the metric, namely
formula (\ref{distancefluc}) with $D+A$ given by (\ref{drelevant}) where
\begin{equation}
 H \neq 0 ,\; A_\mu = 0.
\label{eq:91}
 \end{equation}
This amounts to take the product of the manifold by an internal geometry 
\begin{equation}
T_F^x := (\aa_F, \hh_F, D_F(x) := D_F + H(x))\label{eq:8}
\end{equation}
 in which $D_F$ is now a non-constant
section of $\text{End}\, \hh_F$.
\smallskip

Let $d_{\text{geo}}$, $d_x$, $d_H$ denote  the geodesic distance on $\mm$, the
spectral distance associated to the spectral triple $T_F^x$, and
the distance \eqref{distancefluc} with generalized $1$-form $A$ given by \eqref{eq:91}. 
Corollary \ref{coroll:distancereduite} and proposition \ref{reduction}
gives respectively
\noindent
\begin{prop}\cite[Theo. 2']{Martinetti:2002ij}
\label{Pyth-proj}
For any pure states $\delta_x$, $\delta_y$ of $\cinf$ and  $\omega_F,
\omega_F' \in \pp(\aa_F)$, one has
\begin{eqnarray*}
d_H(\delta_x\otimes\omega_F, \,\delta_x\otimes\omega_F') &=& d_x(\omega_F,\omega_F'),\\
d_H(\delta_x\otimes\omega_F, \,\delta_y\otimes\omega_F) &=& d_{\text{geo}}(x,y).
\end{eqnarray*} 
\end{prop}
\begin{prop}\cite[Theo. 4']{Martinetti:2002ij}\label{kkprop}
Let $\ou,\od$ be two normal pure states of $\aa_F$ with support $s_1,
s_2$ in direct sum, and  such that the sum of their support
commutes with $D_F(x)$ for all $x$. Then
$$d(\delta_x\otimes\ou, \delta_y\otimes\od)= L'((0,x),(1,y)),$$
where $L'$ is the geodesic distance in the manifold $\mm':= [0,1]\times \mm$ equipped with the metric
\begin{equation}
\left(\begin{array}{cc} \norm{\tilde H(x)}^2 & 0 \\
0 & g^{\mu\nu}(x) \end{array}\right)
\label{eq:89}
\end{equation}
in which $g^{\mu\nu}$ is the metric on $\mm$  and  $\tilde H$ is projection on $s_2\aa_F$ of the restriction of $D_H$ to $s_1\aa_F$.
\end{prop}
\begin{figure}[hh]
\begin{center}
\mbox{\rotatebox{0}{\scalebox{.45}{\includegraphics{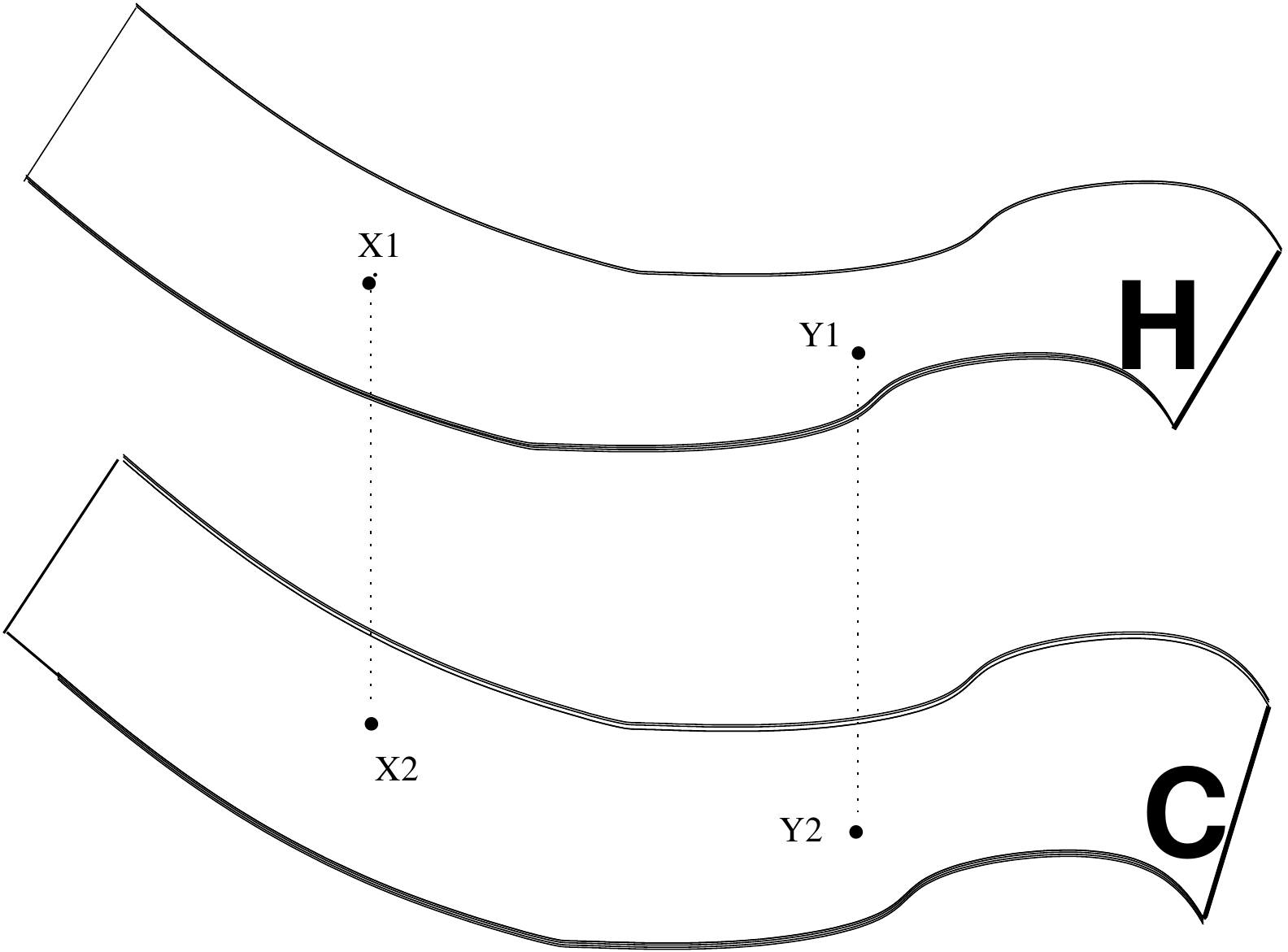}}}}
\end{center}
\caption{ \label{ellipse} Space-time of the standard model with a pure scalar fluctuation of the metric.}
\end{figure}
Proposition \ref{kkprop} gives an intuitive picture of the spacetime
of the standard model. The later is described by an almost commutative
geometry \eqref{eq:88bis} where
\begin{equation}
 \aa_F = \cc \oplus \hhh \oplus M_3(\cc),
\label{eq:99}
 \end{equation}
with $\mathbb H$ the algebra of quaternions. $\A_F$ is suitably represented over
 a finite dimensional vetor space $\hh_F$ generated by elementary fermions, while
 $D_F$ is a finite dimensional matrix that contains the masses of the
 elementary fermions together with the Cabibbo matrix and neutrinos
 mixing angles. We refer the reader to  \cite{Chamseddine:2007oz} for the most advanced
 version of the model pre-discovery of the Higgs, and
 \cite{Chamseddine:2013uq}, \cite{Devastato:2013fk},
 \cite{Devastato:2015aa} for enhanced version post-Higgs.
Through the \emph{spectral action} \cite{Chamseddine:1996kx} the
scalar fluctuation $H$ further identifies with the Higgs field
\cite{Connes:1996fu} (see also \cite{Dubois-Violette:1989fk} for the first appearance of the Higgs
field as a connection in a noncommutative space). 

From the metric point of view, one finds that 
the pure states of $M_3(\cc)$ are at infinite distance from one another,
whereas the states of $\cc$ and $\hhh$ are in direct sum, with support
the identity. Hence the model of spacetime that emerges is a two-sheet model, two copies of the manifold, one indexed by
the pure state of $\cc$, the other one by the pure state of $\hhh$ (cf figure \ref{ellipse}).
\begin{prop}\cite[Prop. 8]{Martinetti:2002ij}
  The distance between the two sheets coincides with the geodesic distance in a $(\text{dim } \mm) +1$ dimension manifold, and the extra-component of the metric is
\begin{equation}
\norm{\tilde H(x)}^2 = \lp\abs{1 + h_1(x)}^2 + \abs{h_2(x)}^2\rp m_t^2\label{eq:24}
\end{equation}
where the $h_i$'s are the components of the Higgs field and $m_t$ is the
mass of the quark~top.
\end{prop}



\section{Sub-Riemannian geometry from gauge fluctuation of the metric}
\label{subriemannian}

In this section we study a gauge fluctuation of the metric, that is 
formula (\ref{distancefluc}) with $D$ substituted with (\ref{drelevant})
where
\begin{equation}
A_\mu\neq 0,\;  H=0.
\label{eq:92}
\end{equation}
This is a review of \cite{Martinetti:2006db} and
\cite{Martinetti:2008hl}. 

As recalled in  \S\ref{sec:subriem}, it was expected that the spectral distance on the bundle of
pure states \eqref{eq:86} were equal to the Carnot-Carath\'eodory (or horizontal) distance associated
to the sub-Riemannian structure defined by the $1$-form field
$A_\mu$. In fact, the link between the
two distances is more intricate and
interesting. As explained in \S\ref{sec:obstruction}, the horizontal
distance is an upper bound to the spectral distance, but it has no
reason to be the lowest one, unless the
holonomy is trivial. In \S \ref{sec:counterexample} we study the
example where the base manifold $\M$ is a circle. The holonomy is not
trivial, and indeed the two distances are not equal. We show it by
working out the connected components of both distances, in case $\A_F=M_2(\C)$. This
result is extended to $\A_F=M_n(\C)$ with $n\geq 2$ in \S\ref{sec:convsacc}.
The two remaining sections contain exact computations of the spectral
distance: on the whole of the bundle of pure states for
$\A_F=M_2(\C)$ (\S \ref{sec:lowdim}); between two pure states on the
same fiber in case $\A_F=M_n(\C)$ for arbitrary $n$  (\S\ref{sec:fiber}). 

 Let us mention that other applications of noncommutative
geometry to sub-Riemannian geometry have been investigated in \cite{Hasselmann:2014fk}.

\subsection{Horizontal structure on the space of pure states}
\label{sec:subriem}
A gauge fluctuation \eqref{eq:92} is obtained from an almost
commutative geometry \eqref{eq:88bis}  by
taking $D_F=0$. In practical, we take as a finite dimensional spectral
triple
\begin{equation}
\aa_F =
M_n(\cc),\quad \hh_F = M_n(\cc),\quad D_F = 0
\label{eq:90}
\end{equation}
for some $n\in\N$, so that the almost commutative geometry we are
dealing with is
\begin{equation}
  \label{eq:93}
  \A =\cinf\otimes M_n(\C),\; \HH = L^2(\M, S)\otimes M_n(\C),\; D= \ds
  \otimes \I_F.
\end{equation}
 The vanishing of $D_F$ implies that the scalar part $H$ of the
fluctuation in \eqref{gaugh} vanishes. Since the spin connection in $\ds$ commutes
with the algebra, the part of the fluctuated operator ~\eqref{drelevant}
relevant in the
distance formula reduces to 
\begin{equation}
  \label{eqbis:24}
  D_\mu:= -i\gamma^\mu (\partial_\mu + A_\mu).
\end{equation}

  As explained in \S \ref{sec:proj}, the space $\pp(M_n(\C))$  of
pure states of  $M_n(\C)$  identifies with the projective space $\cc
P^{n-1}$.
The
 action~\eqref{eq:16} of $SU(n)$ on $\pp(M_n(\C))$ reads as the free
 action of $U(n)$ on $\cc P^{n-1}$,
 \begin{equation}
\xi \to u\xi \quad\quad \forall \xi\in\C P^{n-1}, u\in U(n),
\label{eq:152}
\end{equation}
 and $\pp(\aa)$ in \eqref{eq:86} is now the trivial $SU(n)$-bundle with fiber $\cc P^{n-1}$
\begin{equation}
P\overset{\pi}{\rightarrow}  M\label{eq:11}.
\end{equation}
We denote
 \begin{equation}
\label{defxiox}
\xox :=
 \left(\delta_x\in\pp\left(\cinf\right),\; \omega_\xi\in\pp\left(M_n\left(\C\right)\right)\right)
\end{equation}
  an element of $P$, where $\delta_x$ is the evaluation \eqref{eq:02} and $\ox$ is
  the pure state of $M_n(\C)$ defined by \eqref{eq:45}. Its evaluation on an element of $\A$
  \begin{equation}
 a=f^i\ot m_i,\quad  f^i\in\cinf, m_i\in M_n(\C) 
\label{eq:94}
 \end{equation}
reads
\begin{equation}
\label{evaluationdeux}
 \xi_x(a)= \scl{\xi}{a(x)\xi}, 
 \end{equation}
\smallskip
where for any $x$ in $\M$ one writes
\begin{equation}
a(x) = f^i(x)\otimes m_i\in
M_n(\C).
\label{eq:95}
\end{equation}
The gauge part $A_\mu$ of the fluctuation has value in the set of skew-adjoint elements of $M_n(\C)$, that is the Lie algebra 
$\mathfrak{u}(n)$.
 Thus $A_\mu$ is the local form of the
$1$-form field associated to some  Ehresmann connection $\Xi$ on the trivial
$U(n)$-principal bundle on $\M$. 
By reduction to
$SU(n)$ followed by a mapping to the associated bundle
\eqref{eq:11}, one inherits from $\Xi$ a connection on the bundle
$P$ of pure
states of $\A$. This means that at any $p\in P$ the tangent space $T_pP$ splits into a
vertical subspace and an horizontal subspace,
\begin{equation}
\label{tangenthor}
T_pP = V_pP \oplus H_pP\quad\quad p\in P,
\end{equation} where $HP$ is the kernel of the
connection $1$-form associated to $\Xi$. 

A curve $t\in[0,1]\mapsto
c(t)\in P$  is horizontal when its
tangent vector is everywhere horizontal, that is for any $t$ one has
\begin{equation}
  \label{eq:153}
  \dot c (t) \in H_{c(t)}P.
\end{equation} The horizontal (also called Carnot-Carath\'eodory) distance $d_h(p,q)$ is
defined as the infimum on the length of the horizontal paths joining 
$p$ to $q$,
\begin{equation}
\label{dcc} d_h(p,q) := \underset{\dot{c}(t)\in H_{\!c(t)}\!
P}{\text{Inf}}\; \int_0^1 \norm{\dot{c}(t)} dt \quad \forall
p,q\in P,
\end{equation}
where the norm on $HP$ is the pull back of the metric{\footnote{In all
    this section, $\pi$ denotes the projection from $P$ to $\M$, and not
    the representation of the algebra in the spectral triple.}}
\begin{equation}
\label{normh}
\norm{\dot{c}} = \sqrt{g(\pi_*(\dot{c}),\pi_*(\dot{c}))}.
\end{equation}
When 
$p,q$ cannot 
be linked by any horizontal path then
$d_h(p,q)$
is infinite.

To summarize, the gauge part $A_\mu$ of the covariant Dirac operator (\ref{drelevant})
equips the bundle $P$ of pure states of an almost commutative geometry
with two distances: the
horizontal distance $d_h$ \eqref{dcc} and the fluctuated spectral distance $d_A$
\eqref{distancefluc} computed with $D_\mu$.
 The rest of this section is a collection of results regarding the
comparison of these two distances.

\subsection{Holonomy obstruction}
\label{sec:obstruction}

 \begin{defi}
   A pure state at finite
horizontal distance from $\xox$ is said {\it accessible}, and we define 
\begin{equation}
\label{acc} \text{Acc}(\xox):= \{q\in P;\; d_h(\xox,q) <
+\infty\}.
\end{equation}
A pure state at finite spectral distance from  $\xox$ is said
connected, and we define
\begin{equation}
\label{conn} \text{Con}(\xox):= \{q\in P;\; d_A(\xox,q) <
+\infty\}.
\end{equation}
\end{defi}
\noindent We use the same notation as in \eqref{con} although here we restrict
to pure states.
\smallskip

 In the same way as the spectral distance on a manifold is
bounded by the geodesic distance, for almost commutative geometry with
gauge fluctuation the horizontal distance provides an upper bound to
the spectral distance.
 
\begin{prop}\cite[Prop. 1]{Martinetti:2006db}
 For any $\xox, \yoz\in P$, 
  \begin{equation}
    \label{dinfdh}
    d_A(\xox, \yoz) \leq d_h(\xox, \yoz) \quad \forall \xox,\yoz\in P.
  \end{equation}
In other terms
\begin{equation}
    \label{inclusion} \text{Acc}(\xox)\subset\text{Con}(\xox).
  \end{equation}
\end{prop} 

\noindent However this upper bound is not optimal. In
\cite{Connes:1996fu} was suggested that $d_A$ and $d_h$ were equal. This
is true when the holonomy group reduces to the identity: then
$\text{Acc}(\xox)= \text{Con}(\xox)$ coincides with the horizontal lift
  of $\M$ passing through $\xox$. In particular, on a
given fiber there is no points accessible from one another and both
the spectral and the horizontal distances are infinite.

However when the holonomy is not trivial, then $\text{Acc}(\xox) $ has
no reason to equal $\text{Con}(\xox)$.
The obstruction comes from the number of times a
minimal horizontal curve between $\xox$ and $p\in\text{Acc}(\xox)$
- that is an horizontal curve whose
length is the horizontal distance - intersects the same orbit of the holonomy group. To be more explicit, given an horizontal curve $c$ between
$\xox$ and $\yoz$, we call \emph{ordered self-intersecting points at $p_0=c(t_0)$} a set of $K$ elements
$p_1:= c(t_1),..., p_K:= c(t_K)$  such that for any $i=1,... ,K$
\begin{equation}
\pi(p_i) = \pi(p_0),\quad d_h(p_0, p_{i+1}) >  d_h(p_0, p_i).
\label{eq:125bis}
\end{equation}

\begin{figure}[hbt]
\label{ososipcaption}
\begin{center}
\mbox{\rotatebox{0}{\scalebox{.65}{\includegraphics{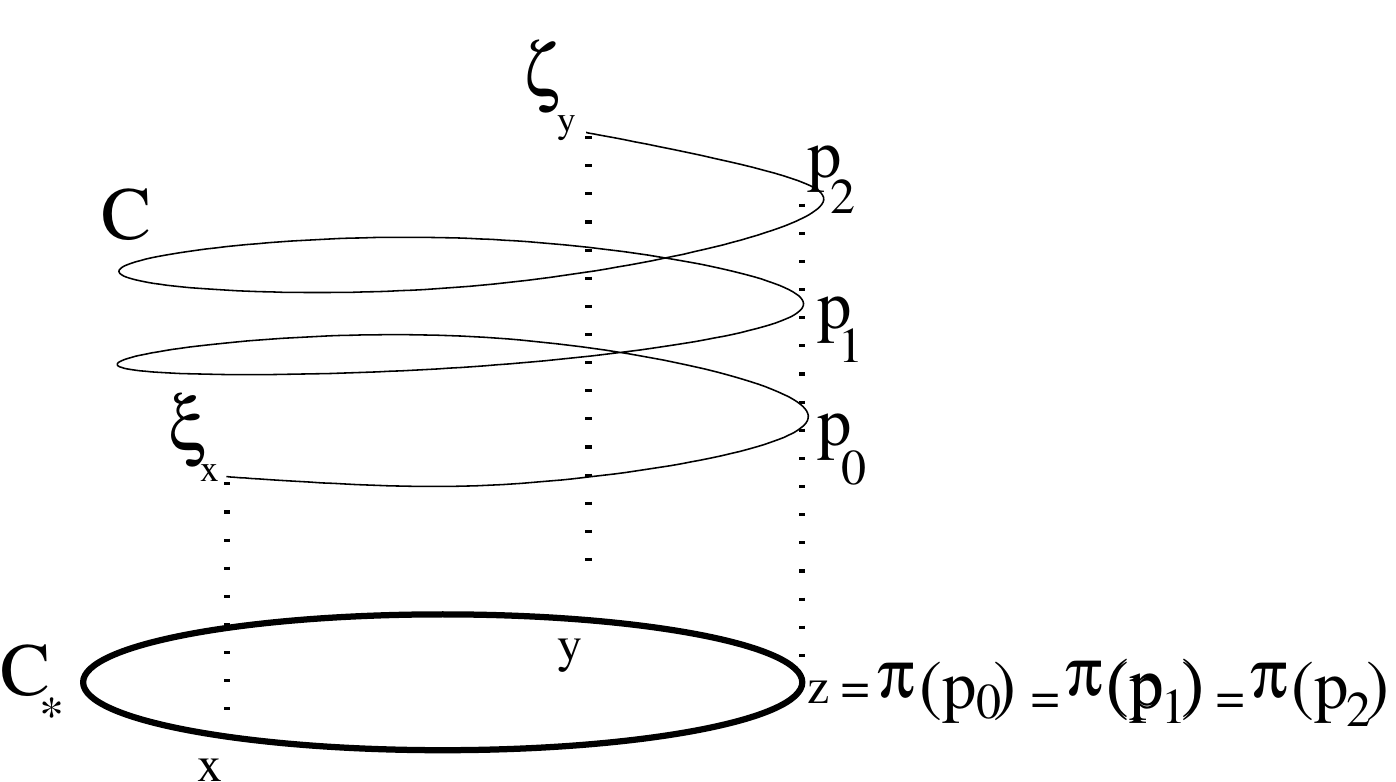}}}}
\end{center}
\caption{An ordered sequence of self-intersecting points.}
\end{figure}
\noindent Assuming the spectral distance between to pure states $\xox$, $\yoz$ is the
horizontal distance, and that there exists at least one minimal horizontal
curve between $\xox$ and $\yoz$, then  one has the following constraint on the optimal
element of definition ~\ref{optimelement}.

\begin{prop} \label{ososip} Let $\xox$, $\yoz$ be two points in $P$ such
that $d_A(\xox,\yoz) = d_h(\xox,\yoz)$.  Then for any minimal horizontal
curve $c$ between $\xox$ and $\yoz$ one has
\begin{equation}
\label{unter} d_A(\xox,c(t))= d_h(\xox,c(t)).
\end{equation} 
Moreover there exists an optimal element $a\in\aa$
such that for any $\xi_{t} := c(t)$
\begin{equation}
\xi_{t}(a) = d_h(\xox, c(t))\;\text{ or }\; \underset{n\rightarrow
\infty}{\text{lim}} \xi_{t}(a_n) = d_h(\xox, c(t)).
\end{equation}
\end{prop}
\noindent Consequently, assuming there is a minimal horizontal curve
between $\xox$ and $\yoz$ with $K$ self-intersecting points at $p_0$, proposition
\ref{ososip} puts $K+1$ condition on the $n^2$ real components of
the selfadjoint matrix $a(\pi(p_0))$,
 \begin{equation}
 \label{holonomycons}
 p_i(a) = \text{Tr}\,(s_{p_i} a(\pi(p_0))) = d_h(\xox, p_i) \quad \forall i=0,
 1, ... ,K
 \end{equation}
 where $s_{p_i}$ denotes the support of the pure state $p_i$. So it is most likely that $d_A(\xox,\yoz)$ cannot equal $d_h(\xox,\yoz)$ 
unless there exists a minimal horizontal curve between $\xox$ and
$\yoz$ such that its projection does not self-intersect more than
$n^2$ times. Actually, questioning the equality between $d_A$ and $d_h$ amounts to the following problem:
 \smallskip

\noindent\emph{Given a minimal horizontal curve $c$, is there a way to
  deform it into another horizontal curve $c'$, keeping its length
  and its end-points fixed, such that $c'$ has less 
selfintersecting points than $c$ ?}
\smallskip

 Say differently:

\smallskip
\noindent \emph{Can one characterize the minimum number of selfintersecting points
  in a minimal horizontal curve between two given points ?}
\smallskip

It seems that there is no known answer to these 
questions \cite{Montgomery:2002yq}. In some cases it might be possible
indeed to  reduce the number of self-intersecting points of
a minimal horizontal curve by
smooth deformations that keep its length constant (see \cite[\S 2.3]{Martinetti:2008hl}).
In order to escape these issues, we  consider a case where there 
is at most one minimal horizontal curve between two points: bundles on the circle $S^1$.
 \bigskip

\subsection{The counter-example of the circle} 
\label{sec:counterexample}
We consider \eqref{eq:93} for $\M = S^1$.
The gauge fluctuation $A_\mu$  has only one component
$A$ and we fix on $\C^n$
a basis of real eigenvectors of $iA$ such that
\begin{equation}
\label{diagA}
A = i\left(
\begin{array}{ccc} 
\theta_{1} &\ldots & 0\\  
\vdots & \ddots&\vdots\\
0&\ldots & \theta_n\end {array}\right),
\end{equation} where the $\theta_j$'s are real
functions on $S^1$. The space of pure states of
\begin{equation}
\A=C^\infty(S^1)\otimes M_n(\C)
\label{eq:101}
\end{equation}
is a $\C
P^{n-1}$ bundle $P\overset{\pi}{\rightarrow} S^1$ on the
circle. In a  trivialization $(\pi, V)$, we associate to the  pure state
$\xox\in P$ with 
\begin{equation}
V(\xox)=\xi = \left(\begin{array}{c} V_1\\\vdots \\
    V_n\end{array}\right) \in \cc P^{n-1},
\label{eq:97}
\end{equation}
the $n-1$-torus of $\cc P^{n-1}$
\begin{equation} \label{txi}
 T_\xi :=  \{ \left(\begin{array}{r} V_1\\
e^{i\varphi_j}V_j \end{array}\right),\; \varphi_j\in \rr,\; j=2,..., n\},
\end{equation}
and the $n$-torus of $P$,
\begin{equation}
\label{tz0}\mathbb{T}_\xi:= S^1 \times T_{\xi}.
\end{equation}

The set $\text{Acc}(\xox)$ of points in $P$ accessible to $\xox$ is
the horizontal lift $c(\tau)$, $\tau\in\R$, of the circle with initial conditions
$\pi(c(0))= x$, $V(c(0))=\xi$.
Explicitly, one has 
\begin{equation}
c(\tau)=(c_*(\tau), V(\tau))
\label{eq:98}
\end{equation}
where $c_*(\tau):=\pi(c(\tau))$ while $V(\tau)$ has components
\begin{equation}
\label{vi} \displaystyle V_j(\tau) = V_j e^{-i\Theta_j(\tau)}\quad
\text{ with }\quad
\Theta_j(\tau) := \int_0^ {\tau}
\theta_j(t)dt.
\end{equation}
Hence on a given fiber $\pi^{-1}(c_*(\tau))$ the set of accessible
points is the sub-torus of $T_\xi$,
\begin{equation}
H^\xi_\tau := \text{Acc}(\xox)\cap \pi^{-1}(c_*(\tau)) =
\{ \left(\begin{array}{r} V_1(\tau)\\
e^{ik\Theta_{1j}(2\pi)}V_j(\tau) \end{array}\right),\; k\in\zz,
j=2,..., n\}.
\label{eq:7bis}
\end{equation}
This is at best dense in $T_{\xi}$  if all the $\Theta_{1j}(2\pi)$'s are distinct and
irrational. The union over all $S^1$ yields
\begin{equation}
\text{Acc}(\xox) = \underset{\tau\in [0,2\pi[}{\bigcup} H_\tau^\xi \;\varsubsetneq \mathbb{T}_\xi.
\label{eq:10}
\end{equation}

The simplest counter-example to the equality between the horizontal and the
spectral distances is given by the $n=2$ case (i.e. $T_\xi = S^1$).
\begin{prop}\cite[Prop. 5]{Martinetti:2006db}, \cite[Prop. 3.4]{Martinetti:2008hl}
\label{propconn2}
  For $\A_F=M_2(\C)$ and a gauge fluctuation $A$ non proportional to
  the identity, one has
  \begin{equation}
\label{eq:104}
\text{Con}(\xox)=\mathbb{T}_\xi.
\end{equation}
\end{prop}

\noindent Thus by \eqref{eq:10} one has that $\text{Acc}(\xox)$ is at best dense in $\text{Con}(\xox)$. Any
element of $\mathbb{T}_\xi$ that is not in $\text{Acc}(\xox)$ is at
finite spectral distance  from $\xox$, although it is infinitely
Carnot-Carath\'eodory far from it.  This shows that the two distances
are not equal.
\smallskip

In this example, the discrepancy between the two distances follows from  the
holonomy obstruction of Prop. \ref{ososip}. The holonomy
is non-trivial because the base $\M=S^1$ is non-simply connected. A
open question is whether there is the same obstruction when the
holonomy comes from the curvature of the connection.

\subsection{Connected versus accessible points on the $\C P^{n-1}$ bundle on $S^1$}
\label{sec:convsacc}
For $n>2$, proposition \ref{propconn2} needs to be refined. $\text{Con}(\xox)$
is still a subset of the torus $\mathbb{T}_\xi$ but not necessary equal to
it.
 Viewing the torus $\mathbb{T}_\xi$ as the subset
of $\rr^n$, 
\begin{equation}
\mathbb{T}_\xi = \left\{ \tau\in [0,2\pi[,\varphi_i \in [0,2\pi[, i=2,...n\right\}
\end{equation} 
one has that $\text{Con}(\xox)$ is a sub-torus
$\mathbb{U}_\xi$ of $\mathbb{T}_\xi$,
\begin{equation}
\mathbb{U}_\xi = \left\{ \tau\in [0,2\pi[,\varphi_i \in [0,2\pi[,i=2,...n_c\right\}
\end{equation}
with dimension $n_c\leq n$ given by the number of equivalence classes of the following
relation. 
 
\begin{definition}
\label{connectdirect}  Let us fix
a pure state $\xox$ in $P$. Two directions $i,j$ of
$\mathbb{T}_\xi$ are said {\it far} from each other if the components $i$
and $j$ of the holonomy at $x$ are equal, and we write
$\text{Far}(.)$ the equivalence classes,
\begin{equation}
\text{Far}(i):= \{ j\in [1,n] \text{ such that } \Theta_j(2\pi)
= \Theta_i(2\pi)\;\text{mod}[2\pi]\}.
\end{equation}
We denote $n_c$ the numbers of such equivalence
classes and we label them as
$$\text{Far}_1 = \text{Far}(1),\,
\text{Far}_p = \text{Far}(j_p)\quad p=2,... ,n_c$$ where $j_p\neq0$
is the smallest integer that does not belong to
$\underset{q=1}{\overset{p-1}{\bigcup}}\text{Far}_q$.
Two directions belonging to distinct equivalence classes are said
\emph{close} to each other.
\end{definition}
 The terminology comes from the following
proposition, which shows that the torus-dimension of the connected components
for the spectral distance is given by the number of directions close
to each other. On the contrary, two directions that are not close to
each other do
not contribute to the connected components: from the spectral distance
point of view, they are infinitely far from each other.


\begin{prop} \cite[Prop. 3.4]{Martinetti:2008hl}
\label{propconneccomp}
$\text{Con}(\xox)$ is the $n_c$ torus
\begin{equation}
\label{defu} \mathbb{U}_\xi:= \underset{\tau\in [0,2\pi[}{\bigcup}
U^\xi_\tau \end{equation} where $U^\xi_\tau\subset T_\xi$ is the
$(n_c\!-\!1)$ torus defined by ($V_{i}(\tau)$ is given in
(\ref{vi}))
\begin{equation}
\label{txif} 
U^\xi_\tau :=  \{ 
\left( 
\begin{array}{rl}
    V_{i}(\tau) & \forall i\in\text{Far}_1\\
    e^{i\varphi_2} V_{i}(\tau) & \forall i\in\text{Far}_2\\ \ldots & \\
    e^{i\varphi_{n_{\!c}} }V_{i}(\tau) & \forall i\in\text{Far}_{n_c} 
\end{array}\right)
, \varphi_j\in\rr, j\in [2,n_c]\}.
\end{equation}
\end{prop}
\medskip

The spectral and the horizontal distances yield two distinct topologies
$\text{Con}$ and $\text{Acc}$ on the bundle of pure states $P$. Obviously
\begin{equation}
e^{i\Theta_{1j}(2k\pi)} = e^{i\Theta_{1i}(2k\pi)}\quad \forall
j\in\text{Far}(i),
\label{eq:102}
\end{equation}
hence $H_\tau^\xi\subset U_\tau^\xi$ fiber-wise and
$\text{Acc} (\xox) \subset \mathbb{U}(\xox)$ globally, as
expected from (\ref{inclusion}). Also obvious is the inclusion of
$\mathbb{U}_\xi$ within $\mathbb{T}_\xi$. To summarize the various
connected components organize as follows,
\begin{equation}
\label{composantesconnexesbis} \text{Acc}(\xox)\subset
\text{Con}(\xox) = {\mathbb{U}}_\xi\subset {\mathbb{T}}_\xi\subset P,
\end{equation}
or fiber-wise
\begin{equation}
\label{composantesconnexes} H_\tau^\xi \subset U_\tau^\xi\subset
T_\xi\subset \cc P^{n-1}.
\end{equation}
The difference between $\text{Acc}(\xox)$ and ${\mathbb{U}}_\xi$ is
governed by the irrationality of the connection, whereas the difference between
 ${\mathbb{U}}_\xi$ and $\mathbb{T}_\xi$ is governed by the number of
close directions. 
 More specifically
\begin{equation}
\label{diverstxi} \mathbb{T}_\xi= \underset{\zeta\in
T_\xi}{\bigcup}\text{Acc}(\zeta_x) \end{equation}
 is the union of all states with equal components up to phase
 factors. Meanwhile
\begin{equation}
\label{diversuxi} \mathbb{U}_\xi= \underset{\zeta\in
U_\xi}{\bigcup}\text{Acc}(\zeta_x), \end{equation}
 with
$U_\xi = U_{\tau = 0}^\xi$, is the union of all states with equal components up
to phase factors, with the extra-condition that phase factors
corresponding to directions far from each other must be equal.

Note that none of the distances is able to "see" between different
tori $\mathbb{T}_\xi$, $\mathbb{T}_\eta$. However within a given
$\mathbb{U}_\xi$ the spectral
distance ``sees'' between the horizontal components.
In this sense the spectral distance keeps "better in mind"
the bundle structure of the set of pure states $P$ (see also figure
\ref{ds} in \S \ref{sec:fiber}).  This suggests
that the spectral 
distance could be relevant to study some transverse metric structure in a more general
framework of foliation.

\subsection{A low dimensional example}
\label{sec:lowdim}
Having individuated the connected components of the spectral distance, we now compute 
the latter explicitly in two examples: on the whole of the bundle $P$
of pure states in the low dimension case $n=2$ below, and on a given fiber for
arbitrary $n$ in \S\ref{sec:fiber}. 
\medskip

Identifying
${\mathcal P}(M_2(\C)) \simeq \C P^{n1-}$ with the $2$-sphere via
\eqref{hopf}, the pure state space of $C^\infty(S^1, M_2(\cc))$ is a
bundle in sphere over $S^1$. The pure state $\xox$ in \eqref{eq:97} is
mapped to the point  
\begin{equation}
  \label{eq:107}
x_0 =  R \cos \theta_0,\quad y_0 =  R \sin \theta_0,\quad z_0 =
z_\xi
\end{equation}
on the fiber $\pi^{-1}(x)$, where we define
\begin{equation}
 2V_1\overline{V_2} =: Re^{i\theta_0}.\label{eq:106}
 \end{equation}
The torus  $T_\xi$ in \eqref{txi} is mapped to the circle of
radius $R$
\begin{equation}
S_R := \left\{ x,y,z \in S^2, z= z_0\right\},
\end{equation}
 so that by Prop. \ref{propconn2} and assuming the holonomy is not
 trivial, the connected component
\begin{equation}
\label{ttxi}
\text{Con}(\xox) = \mathbb{T}_{\xi} = S^1 \times
S_R
\end{equation}
 is a 2-dimensional torus (see Figure \ref{figtore}).

The points accessible from $\xox$ are given in (\ref{eq:98}) as 
\begin{equation}
\xox^k:= c(\tau +2k\pi), \tau\in [0, 2\pi[, k\in \mathbb{Z}.
\end{equation}
On the sphere they 
have coordinates
\begin{equation}
x_\tau^k := R \cos (\theta_0 - \theta_\tau^k),\quad
y_\tau^k:= R \sin (\theta_0 - \theta_\tau^k),\quad
z_\tau^k:= z_\xi
\end{equation}
where
$\theta_\tau^k:= \theta(\tau + 2k\pi).$
${\text{Acc}(\xi_x)}$ is discrete or dense within
$\mathbb{T}_{\xi}$ , depending whether $\Theta(2\pi)$  is rational or
not. 
 \begin{figure}[hh]
\begin{center}
\mbox{\rotatebox{0}{\scalebox{.5}{\includegraphics{{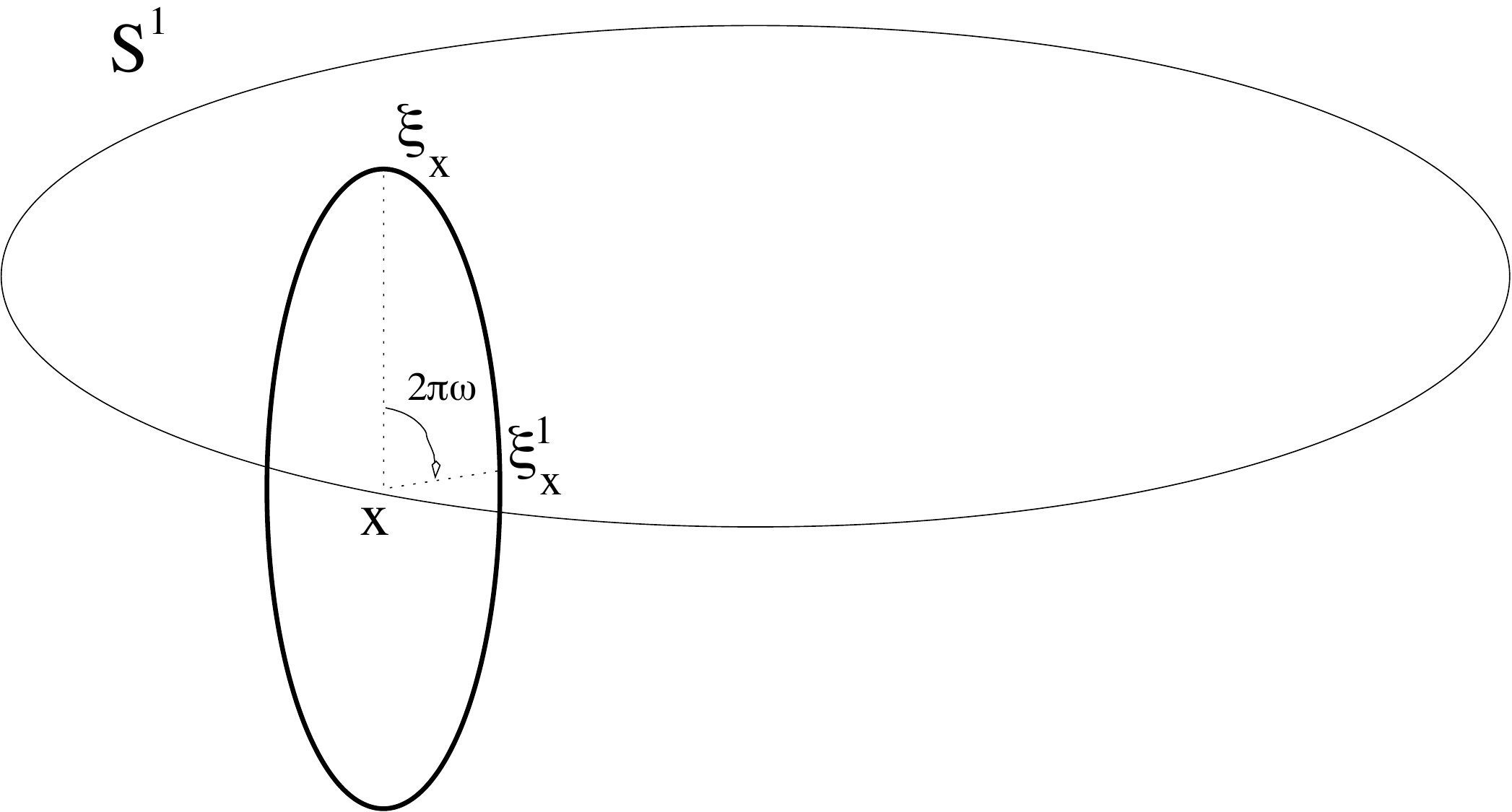}}}}}
\end{center}
\caption{ \label{figtore} The $2$-torus $\mathbb{T}_{\xi}$, the accessible point $\xi_x^1$ and an arbitrary pure state
$\zeta_y$.}
\end{figure}

To compute the spectral distance on on $\mathbb{T}_{\xi}$ we use the
following parametrization.
{\defi{Given $\xox$ in $P$, any pure state $\yoz$ in the
    $2$-torus $\mathbb{T}_\xi$ is in
one-to-one correspondence with an equivalence class
\begin{equation}
\label{parameters} (k\in\nn,\; 0\leq \tau_0\leq 2\pi,\;0\leq\varphi
\leq 2\pi)\sim (k + \zz,\; \tau_0,\;\varphi- 2\zz\omega\pi)
\end{equation} such that
\begin{equation}
\label{tauk} \tau = 2k\pi + \tau_0,\quad \omega:=
\frac{\Theta_1(2\pi) - \Theta_2(2\pi)}{2\pi},\quad
\yoz =
\left(\begin{array}{r} V_1(\tau)\\
e^{i\varphi}V_2(\tau)\end{array}\right).
\end{equation}}}
After a rather lengthy computation, one finds
\begin{prop}\cite[Prop. 4.5]{Martinetti:2008hl}
\label{propcarat} Let $\xox$ be a pure state in $P$ and $\yoz  = (k, \tau_0, \varphi)$
 a pure state in ${\mathbb{T}_\xi}$. Then either
 the two directions are far from each other
so that $\text{Con}(\xox)=\text{Acc}(\xox)$ and
\begin{equation}
\label{caratzero} \displaystyle d_A(\xox, \yoz)= \left\{
\begin{array}{ll}
\text{min} (\tau_0, 2\pi - \tau_0) &\text{when } \varphi = 0\\
+\infty &\text{when } \varphi \neq 0
\end{array}
\right.;
\end{equation}
or the directions are close to each other so that
$\text{Con}(\xox)={\mathbb{T}}_\xi$ and
\begin{equation}
\label{carat} \displaystyle d_A(\xox, \yoz)=
\underset{\mathcal{T}_\pm}{\max} \; H_\xi(T,\Delta)
\end{equation}
where \begin{equation*}
H_\xi(T,\Delta):= T + z_\xi\Delta +
RW_{k+1}\sqrt{(\tau_0-T)^2-\Delta^2}
+RW_{k}\sqrt{(2\pi-\tau_0-T)^2-\Delta^2}
\end{equation*}
with
\begin{equation}
\label{wmaxx} W_k  :=
\frac{\abs{\sin(k\omega\pi+\frac{\varphi}{2})}}{\abs{\sin
\omega\pi}};
\end{equation}
and the maximum is on one of the triangles (see fig. \ref{triangle})
\begin{equation}
  \label{eq:002}
  {\mathcal T}_\pm := T \pm \Delta \leq \text{min}(\tau_0, 2\pi-\tau_0)
\end{equation}
with sign the one of $z_\xi$.
\end{prop}
\begin{figure}[h*]
\label{triangle}
\begin{center}
\mbox{\rotatebox{0}{\scalebox{.5}{\includegraphics{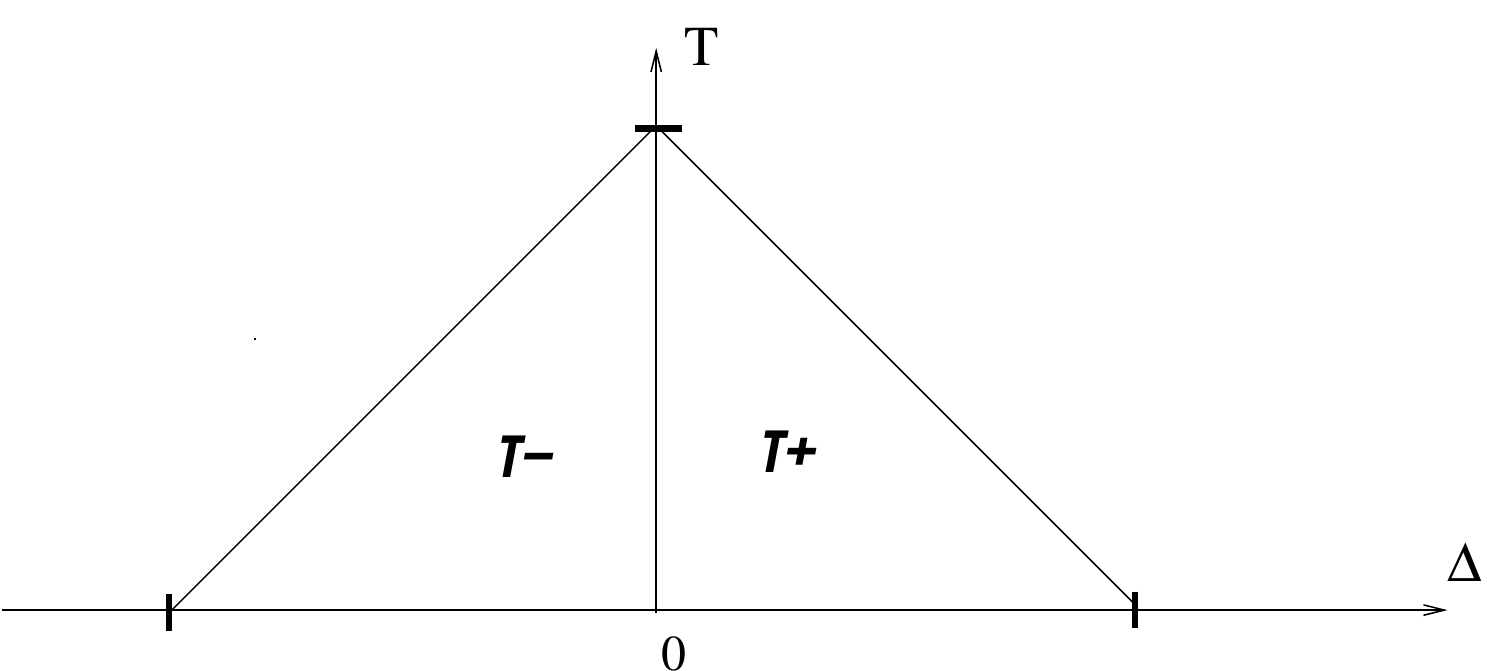}}}}
\end{center}
\caption{Unit is $\text{min}(\tau_0, 2\pi-\tau_0)$. }
\end{figure}
\noindent For $\xi$ an equatorial state, i.e. $z_{\xi} =
0$, the result greatly simplifies
\begin{prop} \cite[\S 4.3.1]{Martinetti:2008hl}
    \label{revisun}
    \begin{equation}
    d_A(\xox,\yoz) = H_\xi(0,0)= RW_{k+1}\tau_0 + RW_{k}(2\pi-\tau_0).
\label{eq:117}
    \end{equation}
\end{prop}
\medskip

\subsection{Distances on the fiber}
\label{sec:fiber}

In the general case $\aa=C^{\infty}(S^1,\mn)$
for arbitrary integer $n\in\nn$ one can explicitly compute the spectral distance for two pure states on the same fiber. 
$\mathbb{T}_\xi$ is now a $n$-torus
and instead of (\ref{parameters}) one deals with equivalence classes
of $(n+1)$-tuples
\begin{equation}
\label{nparameters} \lp k\in \nn,\, 0\leq \tau_0\leq 2\pi, \, 0\leq
\varphi_i \leq 2\pi\rp \sim \lp k+ \zz,\, \tau_0,\, \varphi_j -
2\zz\omega_j\pi\rp
\end{equation}
with
\begin{equation}
\label{ntauk} \omega_j:= \frac{\Theta_1(2\pi) -
\Theta_j(2\pi)}{2\pi}\quad\quad\quad \forall j\in[2,n],
\end{equation}
such that $\yoz$ in $\mathbb{T}_\xi$ writes
\begin{equation}
\label{etattaun}
\yoz = \left(\begin{array}{r} V_1(\tau)\\
e^{i\varphi_j} V_j(\tau)\end{array}\right) \end{equation} where
$\tau:= 2k\pi + \tau_0$. As soon as $n>2$ there is no longer
correspondence between the fiber of $P$ and a sphere, however in
analogy with (\ref{eq:106})  we write
\begin{equation} \label{ndeltaxi}
V_j = \sqrt{\frac{R_j}{2}}e^{i\theta^0_j}
\end{equation}
where $R_j\in\rr^+$ and $\theta^0_j\in[0,2\pi]$. The spectral distance on a given fiber has a simple expression.  To fix notation we consider the fiber over $x$ and we identify
$\xi_x$ to the $n+1$-tuple $(0,0, ..., 0)$.
\begin{prop}\cite[Prop. 5.2]{Martinetti:2008hl}
 Given a pure
  state $\xoz = (k, 0, \varphi_j)\in\mathbb{T}_\xi$, either $\xoz$
  does not belong to the connected component $\mathbb{U}_\xi$ and
  $\label{infn} d(\xox,\xoz) = +\infty $
  or $\xoz\in\mathbb{U}_\xi$ and
  \begin{equation}
    \label{resultatnproof} 
d_A(\xox,\xoz) = \pi \text{Tr}\, {\abs{S_k}}
  \end{equation}
  where $\abs{S_k} = \sqrt{S_k^*S_k}$ and $S_k$ is the matrix with
  components
  \begin{equation}
    \label{defSij} S^k_{ij}:= \sqrt{R_i R_j}\;\frac{\sin\lp
      k\pi(\omega_j - \omega_i)+ \frac{\varphi_j - \varphi_i}{2}\rp}{ \sin
      \pi(\omega_j - \omega_i)}.
  \end{equation}
\end{prop}
\bigskip

In the low dimensional case $n=2$, the connected component
$\text{Con}(\xox)$ of the spectral distance is 
$2$-torus $\mathbb T_\xi$ (see figure \ref{figtore}), whereas
the connected component $\text{Acc}(\xox)$ of the horizontal distance
is a subset of it. Following \eqref{nparameters}, one parametrizes the
$S^1$-fiber of $\mathbb T_\xi$ over $x$ - denoted $S_x$ -  by
\begin{equation}
\Xi:= 2k\omega\pi + \varphi \quad\text{mod} [2\pi]\quad\quad k\in{\mathbb N}, \;0\leq
\varphi\leq 2\pi.
\label{eq:154}
\end{equation}
In this parametrization, $\xox$ has coordinate $\Xi=0$.
\smallskip

\begin{prop} \cite[\S 6.2]{Martinetti:2006db}
\label{dfibre}
For $\zeta_x\in S_x$ with coordinate  $\Xi\in[0, 2\pi]$, one has
\begin{equation}
 d_A(\xox,\zeta_x) 
= \frac{2\pi R}{\abs{\sin
    \omega\pi}}\sin \frac{\Xi}{2}
.\label{eq:155}
\end{equation}
\end{prop}
\bigskip

\noindent It is quite interesting to note that for those points on the fiber which are accessible from $\xox$, namely $\xox^k = (k,0,0)$ or equivalently
\begin{equation}
\Xi = \Xi_k := 2k\omega\pi,
\label{eq:126}
\end{equation}
 the Carnot-Carath\'eodory distance is $d_h(0,\Xi_k) = 2k\pi.$
Hence, as soon as $\omega$ is irrational, one can find close to $\xox$ in the Euclidean topology of $S_x$
some $\xox^k$ which are arbitrarily Carnot-Carath\'eodory-far
from $\xox$. In other terms, $d_h$ destroys the $S^1$ structure of
the fiber. On the contrary the spectral distance preserves it since
$d_A(\xox,\yox)$ is proportional to the chord distance on $S^1$ (see
figure \ref{ds}).
\bigskip

The chord distance already appeared in Prop. \ref{propmercrediun}  for the finite dimensional
spectral triple with $\A=M_2(\C)$. This suggests that
the distance in an almost commutative geometry with a gauge
fluctuation may be retrieved as a spectral
distance associated to the finite dimensional algebra $\A_F$
only, as this happens in the product of spectral triples with the
\emph{non-fluctuated} Dirac operator (see Prop.~\ref{coroll:distancereduite}).
\bigskip

Proposition \ref{dfibre} gives another example where the space of pure
states equipped with the spectral distance is not a path metric
space. We come back to this point in \S\ref{points}. Let us just notice
that the chord distance on $S^1$ is smooth at the cut-locus, contrary to the Euclidean distance
on the circle (cf Figure \ref{ds}). A possibility to make the fiber
$S_x$ equipped with $d_A$ a path metric space, is to interprete the chord
distance as the arc length on a cardioid. We elaborated on this point
in~\cite{Martinetti:2006rz}. 
\vfill

\begin{figure}[h]
\vspace{1.2truecm}\mbox{\rotatebox{0}{\scalebox{.8}{\hspace{-2.4truecm}\includegraphics{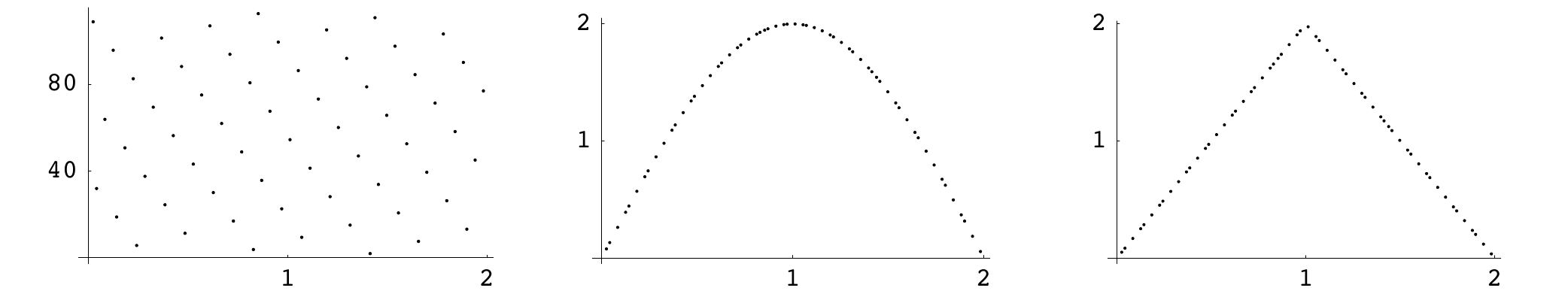}}}}
\caption{ \label{ds} From left to right: $d_h(0,\Xi_k)$, $d_A(0,\Xi_k)$ and
$d_E(0,\Xi_k)$. Vertical unit is $\frac{\pi R}{\abs{\sin\omega\pi}}$,
horizontal unit is $\pi$.}
\vspace{-.25truecm}
\end{figure}

\section{Compact operators: Moyal plane vs quantum space}
\label{sec:compact}
After finite dimensional spectral triples in section \ref{section:finite}
and almost commutative geometries in sections \ref{section:almostcg}
and \ref{subriemannian}, we consider the spectral distance on an
truly noncommutative
algebra (that is: an infinite dimensional algebra with finite dimensional center),
namely the $C^*$- algebra $\K$ of compact operators.
We  first view it as the norm
closure of the algebra of the Moyal plane $({\mathcal S}, \star)$, in
order to have a natural candidate as a Dirac operator. The main results on the spectral distance on the Moyal plane
coming from \cite{Cagnache:2009oe} and \cite{Martinetti:2011fko}
are exposed in \S \ref{section:Moyal}. 

Then we view $\K$ as  the algebra of functions on a quantum
spacetime. This interpretation is common to various models,  motivated by quantum gravity.  Following the idea that at very small scale
spacetime itself becomes quantum, hence the length might be
quantized, a quantum length
operator is defined in these models. It has a priori no links with
the spectral distance but we showed \cite{Martinetti:2011fk} that
between certain classes of states, including coherent states, the
quantum length and the spectral distance actually capture the same
metric information. This is the object of \S \ref{DFR}.

\subsection{Spectral distance on the Moyal plane}
\label{section:Moyal}
The spectral triple of the Moyal plane \cite{Gayral:2004rc} is 
\begin{equation}
  \label{eq:19}
  \A= (\mathcal S, \star), \quad \HH =L^2(\R^2)\otimes \C^2, \quad D =-i\sigma^\mu\nabla_\mu,
\end{equation}
where $\A$ is a noncommutative deformation of the algebra of Schwartz
functions on the plane,
\begin{equation}
  \label{eq:109}
  (f\star g)(x) := \frac 1{(\pi\theta)^2}\int d^2y d^2z f(x+y) g(x+z)
  e^{-2i y \Theta^{-1}z} \quad \forall f,g\in\A,
\end{equation}
 induced by a symplectic form on
$\R^2$,
\begin{equation}
\Theta:=\theta\left(\begin{array}{cc}0&1 \\ -1&0\end{array}\right),\quad
\theta\in(0,1).\label{eq:127}
\end{equation}
$D$ is the usual Dirac operator of the plane and $f\in\A$ acts
$L_f\otimes\I_2$ where 
\begin{equation}
  \label{eq:18}
  (L_f \psi)(x) = (f\star \psi)(x) \quad \forall x\in \R^2
\end{equation}
denotes the left-star multiplication on $L^2(\R^2)$.

The $C^*$-closure of $\A$ is the algebra of compact operators.
Thus the set of pure states of $\A$ is the set of vector states in the
irreducible faithful representation (the Schr\"odinger representation)
of $\A$ on $L^2(\R)$, where any element of $\A$ appears as  an
infinite dimensional matrix with rapidly decreasing coefficients. The eigenstates of the harmonic oscillator, that is the Hermite
functions $h_m$,  form a
basis of $L^2(\R)$. We denote $\omega_m$ the
associated vector state 
\begin{equation}
  \label{eq:20}
  \omega_m = \scl{h_m}{\cdot\; h_m}.
\end{equation}
Another class of interesting states are the coherent states
\begin{equation}
\alpha_\kappa\omega_0:= \omega_0\circ \alpha_\kappa,\quad  \kappa\in\R^2,\label{eq:112}
\end{equation}
which are
the lift to the ground state $\omega_0$ of the action of $\R^2$ on
$\A$ by translation, that is
\begin{equation}
  \label{eq:22}
  (\alpha_\kappa f)(x) = f(x +\kappa).
\end{equation} 
More generally, for any state $\varphi$ of $\A$
we denote its translated by  $\kappa\in\R^2$ as
\begin{equation}
  \label{eq:21}
  \varphi_\kappa = \varphi \circ \alpha_\kappa. 
\end{equation}

The main results on the spectral distance on the Moyal plane are
summarized in the following proposition.
\begin{prop}
\label{propdist}
i.\cite[Theo. 3.7]{ Martinetti:2011fko}  The spectral distance
between any state $\varphi\in\sa$ of the Moyal algebra and any of its
$\kappa$-translated, $\kappa\in\C$, is precisely the amplitude of translation
\begin{equation}
 d(\varphi, \alpha_\kappa\varphi) = \abs{\kappa}.
 \label{eq:131}
 \end{equation}

ii. \cite{Cagnache:2009oe, Martinetti:2011fko} The spectral distance on the Moyal plane takes all possible
 values in $[0,\infty]$. In particular there are pure states at
 infinite distance from one another.
\medskip

 iii. \cite[Prop. 3.6]{Cagnache:2009oe} The distance between the
eigenstates $\omega_m$ of the harmonic oscillator is additive: for
$m\leq n$,
\begin{equation}
d(\omega_m,\omega_n)=\sqrt{\frac{\theta}2}\sum_{k=m+1}^n{{1}\over{{\sqrt{k}}}}.
\label{eq:135}
\end{equation}
\end{prop}

Point ii. follows from i. and the fact that there exist (pure)  states
at infinite distance from one another (see
\cite[Prop. 3.10]{Cagnache:2009oe} and \cite{Cagnache:2009vn}). This an important
difference with the commutative $\theta=0$ case, where the distance
between any pure states - i.e.  any point - is as large as one wants,
but remains finite. 
\medskip

The distance in the Moyal plane computed with the  ``harmonic Dirac
operator'' that appears in various physical models of quantum
spacetime \cite{Grosse:2007ec} has been investigated in
\cite{Wallet:2011uq}. One finds a multiple of the distance computed
with the usual Dirac operator in \eqref{eq:19}, yielding a
formalization of homothetic spectral triples.

\subsection{Minimal length on quantum spacetime}
\label{DFR}

At the Planck scale $\lambda_P$, the general relativistic picture of spacetime as
a smooth manifold $\M$ is expected to loose
any operational meaning, due to the
impossibility of simultaneously measuring with arbitrary accuracy the
four spacetime coordinates $x_\mu$.  This comes  as a consequence of the
\emph{principle of gravitational stability against localization}
\cite{Doplicher:2001fk,Doplicher:2006uq}, which states 
that to prevent the
formation of black-hole during an arbitrarily accurate localization
process, one postulates a non-zero minimal uncertainty in the \emph{simultaneous}
measurement of \emph{all} coordinates of space-time. 
A way to
implement these uncertainty relations is to view the coordinates in a
chart $U$ of 
$\M$ no more as functions $x\in U\subset\M \mapsto x_\mu\in\R$,
but as quantum operators $q_\mu$ satisfying non trivial commutation relations,
\begin{equation}
  \label{eq:7}
  [q_\mu,q_\nu] = i\lambda_P^2 Q_{\mu\nu},
\end{equation}
where the $Q_{\mu\nu}$'s are operators whose properties 
depend on the model. In particular in \cite{Doplicher:1995hc,Bahns:2010fk}, the commutators
$Q_{\mu\nu}$'s are central operators with selfadjoint
closure, covariant under the action of the Poincar\'e group. In this
case,  the quantum coordinates $Q_\mu$
are affiliated to the algebra of compact operators $\K$, in the same way as in the
commutative case the coordinates $x_\mu$ does not belong to
$C_0(\R^n)$ but are affiliated to it. From this perspective, $\K$ plays the role of algebra of noncommutative functions on
the quantum plane. 

A natural candidate to capture the metric information of the quantum
space \eqref{eq:7} is the length operator 
\begin{equation}
  \label{eq:1bis}
  L = \sqrt{\underset{\mu}{\sum} (dq_\mu)^2} \quad\text{ where } \quad dq_\mu := q_\mu\otimes 1 - 1
  \otimes q_\mu.
\end{equation}
The idea is that the minimum $l_P$ of the spectrum of 
$L$
represents the minimum value of the measurable length on a quantum
space \cite{Bahns:2010fk,Amelino-Camelia:2009fk}.
\medskip

Taking advantage of the double role of the algebra of compact
operators as the closure of the Moyal algebra and as the algebra of
functions on quantum spacetime, it is natural to wonder whether the metric information
captured by the length operator $L$ is related to the metric information captured by the
spectral distance $d$ in the Moyal plane. To this aim, given two states
$\varphi, \varphi'$ of $\K$, one needs to associate a number with the length
operator $L$, that could then be compared with $d(\varphi,
\varphi')$. Assuming that $\varphi, \varphi'$ are in the domain of the
operators $Q_\mu$ (that is $\varphi(Q_\mu)$ and $\varphi'(Q_\mu)$ make
sense), the most natural choice is to consider the evaluation of the
separable state $\varphi\otimes \varphi'$ on the operator $L$,
\begin{equation}
  \label{eq:108}
  d_L(\varphi, \varphi'):= (\varphi\otimes\varphi')(L).
\end{equation}
 We
call it the \emph{quantum length} of  the state $\varphi\otimes
\varphi'$. However, to avoid the difficulties in taking the square
root of an operator, it is more convenient to work with the
\emph{quantum square length}
\begin{equation}
  \label{eq:110}
  \sqrt{d_L^2(\varphi, \varphi')}:= \sqrt {\varphi\otimes \varphi'(L^2)}.
\end{equation}

In the commutative case $q_\mu=x_\mu$, one has
\begin{equation}
  \label{eq:111}
  d_L(\delta_x, \delta_y) = \sqrt{d_{L^2}(\delta_x, \delta_y)} =
  d(\delta_x, \delta_y)  = d_\text{geo}(x, y).
\end{equation}
All the tools introduced so far to measure a distance (spectral
distance, spectrum of a length operator or of its square) all coincide
with the geodesic distance. In the noncommutative case,  there is no reason that the three
quantities on the l.h.s. of the equation above remain equal. In
particular while the spectral distance $d$ is actually a distance between
states, so that $d(\varphi, \varphi)=0$,  there is no reason that
$d_{L^2}(\varphi, \varphi)$ and $d_L(\varphi, \varphi)$ vanish. 

This
can be checked on the set of \emph{generalized coherent states}
(see \eqref{eq:112})
\begin{equation}
  \label{eq:113}
  {\mathcal C} :=\left\{ \alpha_\kappa \omega_m, m\in\N, \kappa\in \R^2\right\}. 
\end{equation}

 \begin{prop} \textnormal{\cite{Martinetti:2011fk,
  Bahns:2010fk}}
\label{propquantdist} The quantum square-length on ${\mathcal C}$ is
\begin{equation}
  \label{eq:bbb162}
  d_{L^2}(\alpha_{\kappa} \omega_m, \alpha_{\tilde\kappa} \omega_n) =
  2E_m + 2E_n + \abs{\kappa-\tilde\kappa}^2
\end{equation}
for any $m,n\in \N,\; \kappa, \tilde\kappa\in\R^2$, with
\begin{equation}
E_m = \lambda_P^2(m+\frac 12)\label{eq:162}
\end{equation}
 the $n^\text{th}$ eigenvalue of the
 Hamiltonian $H$ of the quantum harmonic oscillator.
Hence the quantum square length is invariant by translation. Moreover one has \begin{equation}
\label{eq:840}
  d_L(\alpha_\kappa\omega_m,\alpha_{\tilde\kappa}\omega_n) \leq  \sqrt{d_{L^2}(\alpha_\kappa\omega_m,\alpha_{\tilde\kappa}\omega_n)} 
\end{equation}
with equality only when $m=n=0$ and $\kappa=\tilde\kappa$.
\end{prop}
\noindent Identifying the parameter $\theta$ in
\eqref{eq:127} with the square $\lambda_P^2$ of the Planck length, one
has  between eigenstates of the harmonic oscillatore that
\begin{equation}
\label{discrep1}
d(\omega_m,\omega_n)=\frac{\lambda_P}{\sqrt 2}\sum_{k=m+1}^n{{1}\over{{\sqrt{k}}}}
\; \quad\neq \quad\sqrt{d_{L^2}(\omega_m,\omega_n)} =\sqrt{2E_m + 2E_n},
\end{equation}
whereas between generalized coherent states
\begin{equation}
\label{discrep2}
 d(\omega_m, \alpha_\kappa\omega_m) = \abs{\kappa} \quad
\neq \quad \sqrt{d_{L^2}(\omega_m, \alpha_{\kappa} \omega_m)}
= \sqrt{4E_m + \abs{\kappa}^2}.
\end{equation}

To understand the discrepancy between the quantum length and the
 spectral distance, one should understand how to \emph{turn the quantum length into a true distance} that vanishes
on the diagonal $\varphi'=\varphi$, or how to
\emph{give a quantum taste to the spectral
distance} so that it no longer vanishes on the diagonal. As explained
below, the two points of view turn out to be equivalent thanks to
Pythagoras theorem of \S \ref{section:pythagoras}.
\medskip

One ``quantizes'' the spectral distance by doubling the Moyal plane, that
is by taking the product in the sense of \eqref{eq:78}  of the spectral
triple \eqref{eq:19} with the
two point space of~\S~\ref{sec:discretespaces}, namely
\begin{equation}
\label{eq:9}
\A'\,:=  \A\otimes\, \C^2,\quad\hh':= \hh\,\otimes\,\C^2,\quad
D':= D\,\otimes\,\ii +
\Gamma\otimes D_F
\end{equation}
where $\Gamma$ is a grading of $\hh$ and 
\begin{equation}
D_F:=  \left(\begin{array}{cc} 0&\bar\Lambda\\
\Lambda&0 \end{array}\right)\;\text{ with } \;\Lambda =
\text{ const}.\label{eq:78bis}
\end{equation}
Pure states of $\A'$ are pairs
\begin{equation}
\omega^i := (\omega,\delta_i), \quad 
\omega\in\pa, \; {\mathcal P}(\C^2) = \left\{\delta_1, \delta_2 \right\}.
\label{eq:114}
\end{equation}
 Hence
 \begin{equation}
\pp(\A') \simeq \pa \times \pa\label{eq:115}
\end{equation}
and the geometry described by the doubled spectral triple  $(\A',
\hh', D')$   is a \emph{two-sheet model} - two copies of the Moyal
plane. The 
associated distance $d'$ is known between translated states $\varphi^i=
(\varphi, \delta^i)$, $\varphi_\kappa^j= (\varphi_\kappa, \delta^j)$
non-necessarily localized on the same copy, 
and is given by Pythagoras theorem.
\begin{prop}\cite{Martinetti:2011fko}
\label{Pythquant}
For any $\varphi\in\sa$ and $\kappa\in \R^2$, one has
\begin{equation}
  \label{eq:10bis}
  {d'}^2(\varphi^1, \varphi_\kappa^2) =
  d^2(\varphi, \varphi_\kappa) +
  d_2^2(\delta^1, \delta^2)
\end{equation}
where $d$, $d_2$ are the distances on the Moyal plane and the two point space.\end{prop}
Rather than comparing the quantum length
  with the spectral distance on a single sheet, the idea 
is to 
  compare the quantum square-length $\sqrt{d_{L^2}(\varphi, \tilde\varphi)}$
  with the spectral distance in the double-sheeted model  $d'(\varphi^1,
\tilde \varphi^2)$.
In particular, 
for $\tilde\varphi =\varphi_\kappa$,  one has from
Prop. \ref{Pythquant} and \eqref{eq:28} that
\begin{equation}
d'(\varphi^1, \varphi_\kappa^2 ) = 
\sqrt{d_{L^2}(\varphi, \varphi_\kappa)} 
\label{eq:68new}
\end{equation}
if and only if 
on a single sheet
\begin{equation}
\label{identific}
d(\varphi, \varphi_\kappa) 
=   \sqrt{d_{L^2}(\varphi, \varphi_\kappa) -\abs{\Lambda}^{-2}}.
\end{equation}
 For any $\varphi, \tilde\varphi$ in the domain of the length
  operator $L$, we thus  define the \emph{modified quantum length} as
  \begin{equation}
    d'_L(\varphi, \tilde\varphi) := \sqrt{\abs{d_{L^2}(\varphi, \tilde\varphi)
        - \Lambda^{-2}(\varphi, \tilde\varphi)}}
    \label{eq:83}
  \end{equation}
  where
  \begin{equation}
    \Lambda^{-2}(\varphi, \tilde\varphi) := \sqrt{d_{L^2}(\varphi,\varphi)
      d_{L^2}(\tilde\varphi, \tilde\varphi)}.\label{eq:85}
  \end{equation}
\noindent This is the correct quantity, built from the
length operator $L$, that should be compared with the spectral distance.
\begin{prop}\cite{Martinetti:2011fk}
\label{propmodifiedql}
On the set of generalized coherent states for a fixed $m\in\N$, that is
\begin{equation}
{\mathcal C}(\omega_m):=\left\{\alpha_\kappa\omega_m,
  \kappa\in\R^2\right\},
\label{eq:40}
\end{equation}
one has
\begin{equation}
d_D(\omega,  \tilde\omega)  = d'_L(\omega,  \tilde\omega) \quad\quad\quad
\forall \omega, \tilde\omega\in {\mathcal C}(\omega_m).
\label{eq:69}
\end{equation} 

On the set of all generalized coherent states \eqref{eq:113}, $d_D$ coincides with
$d'_L$ asymptotically, both in the limit of large translation

\smallskip\begin{equation}
  \label{eq:185quinte}
  \lim_{\kappa\to \infty} \frac{d_D(\akom, \akton) -d'_L(\akom,
    \akton)}{d'_L(\akom,\akton)} = 0, \quad \forall m,n\in\N,\, \tilde\kappa\in\C,
\end{equation}
\smallskip

 and for large difference of energy

\smallskip\begin{equation}
  \label{eq:185}
  \lim_{n\to 0} \frac{d_D(\akom, \akton) -d'_L(\akom,
    \akton)}{d'_L(\akom,\akton)} = 0,\quad \forall m\in\N,\, \kappa,\tilde\kappa\in\C.
\end{equation}
\end{prop}
\bigskip

It is quite remarkable that the spectral distance $d_D$ on a single copy of
the Moyal plane coincides (exactly on the set  of translated of a
states, asymptotically on the set of generalized coherent states) with
the ``natural'' quantity $d'_L$, vanishing on
the diagonal, that one can
build from the quantum length $d_L$. The two options
``quantizing the spectral distance'' by allowing the emergence of a
non-zero minimal spectral distance,  or ``geometrizing the quantum
length'' by turning it into a true distance are two equivalent
procedures.
\newpage
\section{Discussion}
\label{sec:discussion}

In \S\ref{points} and \S\ref{sec:geodesics} we gather several observations coming from the previous examples concerning
the following question: what should play the role of points and
geodesics in noncommutative geometry ? As a concluding remark, we
discuss in \S\ref{sectionncg} Kantorovich duality in the noncommutative framework.

\subsection{Points and geodesics}
\label{points}
In the commutative case $\A=C(\M)$ for $\M$ a compact manifold, two pure
states $\delta_x$, $\delta_y$ provide via the GNS construction two
inequivalent irreducible
representations. This is no longer true in the noncommmutative
case. For instance all the pure states of $M_n(\C)$ yields equivalent
irreducible representations. A point of view is to consider that a
``point'' in noncommutative geometry should be a class of irreducible
equivalent representations - that is a class of pure states - rather than a pure
state. From this point of view, any spectral triple with algebra
$\A=M_n(\C)$ describes a one-point space. On the contrary, we argued in
\S\ref{sec:proj} that the spectral distance gives a
non-trivial structure to the set of pure states of $M_n(\C)$,
regardless the unitary equivalence of the representations they induce. We
dot not see any good reason to wash out this structure by considering only
quotients of $\pa$ instead of $\pa$ entirely.

Furthermore, several facts suggest that the purity of state might
not be such a relevant concept regarding the metric aspect of
noncommutative geometry. For instance Pythagoras equality holds between pure states
$\delta_x^1, \delta_y^2$ in the product of a manifold by $\C^2$, but in the case of the Moyal
plane it holds between translated states $\varphi^1, \varphi_\kappa^2$, pure or not. What is
important to pass from Pythagoras inequalities of
theorem \ref{thm:Pythagoras} to the equality is not
the purity of the states, but the existence of a curve $t\to
\varphi(t) \in \sa$ between the two considered states such that 
\begin{equation}
  \label{eq:118}
  d(\varphi(s), \varphi(t)) = |s-t|\, d(\varphi(0), \varphi(1)) \quad
  \forall s,t\in[0,1].
\end{equation}
In case of a manifold, such a  curve is provided by the minimal
geodesic between $\delta_x=\varphi(0)$ and $\delta_y=\varphi(1)$, which has value in pure states. In case of the Moyal plane,
this curve is the orbit of $\varphi=\varphi(0)$ under the translation action
of $\R^2$, which lies in $\pa$ if $\varphi$ is pure, in $\sa$
otherwise. 

Other instances where the purity of state does not seem an adequate
criteria to characterize a ``point'' of a noncommutative geometry are the cut-off geometries
developed in \cite{DAndrea:2013kx}. There,  pure states need to be approximated by
non-pure states. Namely, given a spectral triple $(\A, \HH, D),$ one truncates the Dirac operator via the
adjoint action of a sequence of increasing projections $P_N$ tending to
$\I$,
\begin{equation}
  \label{eq:119}
  D\to D_N:= P_N D P_N.
\end{equation}
This has motivations from the spectral action where the Dirac operator is truncated by a cut-off energy. In
case $(\A, \HH, D)$ is the usual spectral triple of  a manifold,
substituting in the spectral distance formula the semi-norm $L_D$ by 
\begin{equation}
  \label{eq:120}
  L_N(a) := ||[D_N, a]||
\end{equation}
yields a distance between any pure states $\delta_x, \delta_y$ which is
infinite as soon as $D_N$ has finite rank. To make it finite, one
should truncate the pure states as well. In case $\M=S^1$, an explicit
truncation is given by the Fejer transform of rank $N$, yielding non
pure states of $\cinf$ (see \cite[\S
5]{DAndrea:2013kx} for details).
\medskip

Related to the problem of determining what
points are in a noncommutative context, is the question of what should
play the role of a geodesic. From a purely metric point of view, one
may take as a definition of (minimal) geodesic between two
states $\varphi, \varphi'$ a curve like \eqref{eq:118} with
$\varphi(0)=\varphi$, $\varphi(1)=\varphi'$. In the case of a manifold
\begin{equation}
\A=\cinf,\quad  \varphi =\delta_x,\; \varphi'=\delta_y
\label{eq:130}
\end{equation}
there are two such curves:
the usual geodesic between $x$ and $y$ (with $t$ the proper length)
which lies completely in $\pa$, and the convex combination
\begin{equation}
\varphi(t) = t\delta_x + (1-t)\delta_y\label{eq:134}
\end{equation}
which lies in non-pure
states.
In case of $\A= M_2(\C)$, the distance on the $2$-sphere $\pa$  is
the Euclidean distance in the $3$-ball $\sa$, hence any curve \eqref{eq:118}
between two pure states necessarily goes through
non-pure states. The same is true for the two point space of \S
\ref{sec:discretespaces}, the distance between the two sheets of the
standard model in \S \ref{higgs}, and the distance on a fiber $S_x$ with the gauge
fluctuated Dirac operator in \S \ref{sec:fiber}. In other terms, in
all these examples the space of pure states $\pa$ equipped with the spectral
distance is not a path metric space. This forbids to take \eqref{eq:118} as a
definition of a geodesic, at  least as long as one imposes that the later must
be a curve of pure states.

 If one  allows non-pure states, then  a
geodesic in the sense of \eqref{eq:118}  always exists and is given by
\eqref{eq:134}.
 This does
not seem a very operative definition of a geodesic; it simply shows that
by the very definition of states as convex combinations of pure states,
then the space of states $\sa$ equipped with the spectral distance is
 trivially always path metric. A more interesting question could be
 the following: is the commutative case $\A=\cinf$ the
only example where the space of pure states equipped with the spectral
distance is path metric ? 

More understanding on these questions may come from optimal
transport.  As pointed out by a referee of an early version of this
text, the question in this context  is whether the curve of measures (a curve of states, in our terminology) is
produced by the underlying measure on curves. A discussion on that
matter, for the Wasserstein distance of order $p$ though, can be found
in \cite{Lisini:2007aa}, see also \cite{Villani:2009tp}. For the distance of order $1$,
one should see the appendix of \cite{Paolini:2012aa}.
\medskip


\subsection{Optimal elements and geodesics}
\label{sec:geodesics}
Another point of view \cite{Martinetti:2011fk,Martinetti:2011fkbis} on
the question of geodesics in noncommutative geometry could be to define a geodesic in a dual
way, that is to find a substitute of the geodesic in the notion
of \emph{optimal element} introduced in definition
\ref{optimelement}. This makes sense because in the commutative case, the commutator norm
condition
\begin{equation}
\norm{[\ds,f]} = \suup{x\in\M}{\norm{{\nabla f}_{\lvert
    x}}_{T_x \M}} = 1
\label{eq:58bisbisbis}
 \end{equation}
characterizes the optimal element between $\delta_x$ and $\delta_y$ locally, in the
 sense that the constraint is carried by the gradient of $f$. The geodesics through $x$
are retrieved as the curves tangent to gradient of the optimal element
$f=d_{\text{geo}}(x,\, .)$.
In this sense, computing the
spectral distance - that is finding an optimal element - amounts to solving the equation of the geodesics:
\begin{enumerate}
\item[-] eq. (\ref{eq:58bisbisbis}) plays the role of the geodesic equation;

 \item[-] the optimal element $d_{\text{geo}}(x,\, .)$ fully characterizes the geodesics through $x$;

\item[-] the valuation of the optimal element on $\delta_x - \delta_y$ gives the integration of
  the line element on a minimal geodesic between $x$ and $y$.
\end{enumerate}
 \smallskip

At the moment there is no clear translation of the above points in a
noncommutative context. However, focusing on the optimal element
yields interesting interpretations of the results on the Moyal plane
of section \ref{sec:compact}.
 Recall that in the commutative case, as observed in \eqref{eq:111}, both the quantum length and the spectral distance
coincide with the geodesic distance; and this is the same
function
\begin{equation}
  \label{eq:123}
  l(x_\mu):= \sqrt{\sum_\mu x_\mu^2}
\end{equation}
which yields both the optimal element between two pure states $\delta_x,
\delta_{\lambda x}$, $
\lambda \in \R^+$
and - by the functional calculus -  the length operator $L=l(dq_\mu)$
in \eqref{eq:1bis}.  On the contrary, on the Moyal plane the quantum length and the spectral distance
no longer coincide,  as stressed in \eqref{discrep1} and
Prop.~\ref{propmodifiedql}, so that  one could expect the length
operator not to
be defined by the same function as the optimal element. This is indeed
the case. To see
it is convenient to work with the
creation/annihilation operators 
\begin{equation}
  \label{eq:124bis}
  a:= \frac{q_1 + i q_2}{\sqrt 2},\quad   a^*:= \frac{q_1 - i q_2}{\sqrt 2},
\end{equation}
as well as with their universal differentials
\begin{equation}
  \label{eq:96bis}
 da = \frac 1{\sqrt 2}(dq_1 + idq_2), \quad  da^* = \frac 1{\sqrt 2}(dq_1 - idq_2).  
\end{equation}
 \begin{prop}\cite{Martinetti:2011fk}
\label{moyalgeo}
The length operator can be equivalently
defined as  
$L=l_i (da)$, 
with
\begin{equation}
l_1(z)
:=  \sqrt{z \bar z + z\bar z} 
\;\text{ or } \; l_2(z) :=  \sqrt{2(z \bar z  - \lambda_P^2)} \;\text{ or } \;l_3(z)
:=  \sqrt{2(\bar z z + \lambda_P^2)}. 
\label{eq:158}
\end{equation}
The optimal element between any two eigenstates of the Hamiltonian of the
quantum harmonic oscillator is - up to regularization at infinity
- the $\star$-action of  the function $l_0$, defined as the solution of 
\begin{equation}
  \label{eq:1590}
  \left(\partial_z l_0 \star z\right) \star \left(\partial_z
    l_0 \star z\right)^*= \frac 1{2} z^*\star z.
\end{equation} 
Neither $l_{1}(a)$ nor $l_{2}(a)$ or $l_{3}(a)$ are optimal elements
between eigenstates.
\end{prop}


In a similar way, the
spectral distance between 
translated states $\varphi$, $\varphi_\kappa$ 
being the
 amplitude of translation $|\kappa|$ both in the commutative and the
 noncommutative cases, one could expect the respective optimal
 elements to be related. And this is indeed the same function 
\begin{equation}
  \label{eq:133}
  l_\kappa(z) = \frac{ze^{-i\Xi} + \bar ze^{i\Xi}}{\sqrt 2}\quad \text{ with
  } \quad \Xi := \text{Arg} \,\kappa,
\end{equation}
which yields the optimal element (up to regularization at infinity) both on the Euclidean plane (through the pointwise
action of $l_\kappa$) and the Moyal plane (through its $\star$-action). For the
latter, this has been shown in \cite[Theo. III.9]
{Martinetti:2011fko}, for the former in
\cite[Prop. 3.2]{dAndrea:2009xr}). It is quite remarkable that the same function $l_\kappa$ gives an optimal element
between translated states, regardless of the commutativity of the
algebra. 
\smallskip

Let us now compare the optimal elements
$l_\kappa$ for translated states and $l_0$
for eigenstates of the
harmonic oscillator. Modulo regularization at infinity, the latter is characterized as a solution
of \cite[Prop. 3.7]{Cagnache:2009oe}
\begin{equation}
  \label{eq:143bis}
  [\ds,L_{l_0}] =-i\left(\begin{array}{cc} 0 & S^*\\ S&0\end{array}\right)
\end{equation} 
where $L_{l_0}$ denotes the $\star$-multiplication by $l_0$ defined by
\eqref{eq:1590}, while $S$ is the shift operator (eq.~(\ref{eq:1590}) actually follows
from \eqref{eq:143bis}).  In analogy with the
commutative case where $[\ds, f] = (\slashed \partial f)$, we
interpret $[\ds, L_{l_0}]$ as the derivative of the optimal element
$l_0$. The presence of the shift operator in this derivative suggests
that the ``geodesic'' is somehow non smooth. A similar interpretation
follows from the observation that the spectral distance \eqref{eq:135}
\begin{equation}
  \label{eq:125}
    d(\omega_m,
\omega_n) =\lambda_P \sum_{k=m+1}^n{{1}\over{{\sqrt{2k}}}}
\end{equation}
is the middle Riemann sum
approximation of the modified quantum length \eqref{eq:83}
\begin{equation}
  \label{eq:124}
 d_L'(\omega_m,\omega_n) = \lambda_P\left(\sqrt{2n+1} -
  \sqrt{2m+1}\right) = \lambda_P\int_{m+ \frac 12}^{n+\frac 12} \frac 1{\sqrt{2k}}dk.
\end{equation}
In \cite{Martinetti:2011fko,
  Martinetti:2011fk}  we interpret this result saying
that the spectral distance and the quantum length are the integration
of the same quantum line
element
\begin{equation}
\lambda_P\frac 1{\sqrt{2k}}dk
\label{eq:87bis}
\end{equation}
but along two distinct geodesics: a
continuous one for the quantum length \eqref{eq:124}, a discrete one for the spectral
distance \eqref{eq:125}.

Between translated states, the optimal element $l_\kappa$ satisfies an
equation similar to \eqref{eq:143bis}, 
\begin{equation}
   \label{eq:143}
   [\ds, L_{l_\kappa} ] = -i\left(\begin{array}{cc} 0 & e^{i\Xi}\\ e^{-i\Xi}&0\end{array}\right)
 \end{equation} 
where the shift is substituted with a term proportional to the
identity. This indicates that
the  geodesic is ``smooth'', in agreement with the analysis developed
below~\eqref{eq:133}. 

In the same vein, one has
\begin{equation}
  \label{eq:880}
[\ds, L_{l_0}]^*[\ds, L_{l_0}] =  \ii - e_0
\end{equation}
where $e_0$ is the projection on the ground state $h_0$, while
\begin{equation}
  \label{eq:88ter}
  [\ds,L_{l_\kappa}]^*[\ds, L_{l_\kappa}] = \ii.
\end{equation}
Eq. \eqref{eq:88ter} indicates that the derivative of the optimal
element $l_\kappa$ is a unitary operator, whereas the derivative of
$l_0$ is not. This comes from the fact that the set of eigenstates of the harmonic
oscillator - identified to $\N$ - is not a group, unlike
the set of translated states. So the shift $S$ acting on $l^2(\N)$ is not a
unitary operator. 

\subsection{Kantorovich duality in noncommutative geometry ?}
\label{sectionncg}

The formula \eqref{eq:1} of the spectral distance
is a way to export to the noncommutative setting the usual notion of
Riemannian geodesic distance. Notice the change of
point of view: the distance is no longer the infimum of a
geometrical object (i.e. the length of the
paths between points), but the supremum of an algebraic quantity
(the difference of the valuation of two states). 

A natural question is whether one
looses any trace of the distance-as-an-infimum  by passing to the
noncommutative side. More specifically, is there some
``noncommutative Kantorovich duality'' allowing to view the spectral
distance as the minimization of some ``noncommutative cost'' ? 
\begin{align*}
\text{\underline{distance as a supremum}:} & \hspace{1truecm} d_{\ds}  \text{ commutative case}
&\rightarrow&\quad d_D \text{ noncommutative case}\\
&\hspace{1truecm} \uparrow & & \quad\quad\;\lvert \\
 \text{Kantorovich duality: } &\hspace{1truecm} d_{\ds} = W& &\quad d_D = W_D  ? \\
&\hspace{1truecm} \downarrow & & \quad\quad\downarrow\\
\text{\underline{distance as an infimum}:}&\hspace{1truecm}  W \text{
  with cost } d_{\text{geo}}& & \quad\text{noncommutative cost ?}\\
\end{align*}
In this diagram, $d_{\ds}$ and $d_D$ denote the spectral distances computed with the
seminorms $||[\ds, \cdot]||$ and $||[D, \cdot]||$, while $W$ is the
Wasserstein distance and $W_D$ its putative noncommutative generalization.

In the commutative case, the cost function is retrieved as
the Monge-Kantorovich distance between pure states of $C_0(\M)$. So in
the noncommutative case, if
the spectral distance were to coincide with some
``Monge-Kantorovich''-like distance $W_D$ on $\sa$,
then the associated cost should be the spectral distance on the pure
state space $\pa$. So given a spectral triple 
$(\A, \hh, D)$, we aim at defining a
``Monge-Kantorovich''-like distance $W_D$ on the state space $\sa$, taking as a cost function the
spectral distance $d_D$ on
the pure state space $\pa$. A first idea is to mimic formula
\eqref{eq:3} with $\X = \pa$, that is
\begin{equation}
  \label{eq:30}
W(\mu_1, \mu_2) = \inf_{\rho} \int_{\pa\times\pa} \hspace{-1.25truecm} \de\rho \quad d_D(\omega, \tilde\omega)
\end{equation}
where $\mu_1, \mu_2$ are probability measures on $\pa$, $\omega,
\tilde\omega$ are generic elements of $\pa$  and the infimum is on the
measures $\rho$  on $\pa\times \pa$ with marginals $\mu_1$,
$\mu_2$. For this to make sense as a distance on $\sa$, we should restrict to states
$\varphi\in\sa$ that are given by a probability
measure on $\pa$.  This is possible (at least) when $\A$ is separable and
unital: $\sa$ is then metrizable
\cite[\S 4.1.4]{Bratteli:1987fk} so that by Choquet theorem any state
$\varphi\in\sa$ is given by a probability measure
$\mu\in\text{Prob}(\pa)$. One should be careful however that the correspondence is not $1$ to $1$:
$\sa\to\text{Prob}(\pa)$ is injective, but two distinct probability
measures $\mu_1, \mu_2$ may yield the same state $\varphi$. This
is because $\A$ is \emph{not} an algebra of continuous functions on
$\pa$ (otherwise $\A$ would be commutative). 
Thus $W_D$ that we are looking for should not be a distance on  $\text{Prob}(\pa)$, but on a quotient of
it, precisely given by $\sa$. This forbids to define $W_D$ by
formula (\ref{eq:30}), since by construction the latter is a distance
on  $\text{Prob}(\pa)$.

A possibility is to consider the infimum 
\begin{equation}
\inf_{\mu_1,  \mu_2} \; W(\mu_1, \mu_2)
\label{eq:33}
\end{equation}
on all the probability measures $\mu_1,
\mu_2\in\text{Prob}(\pa)$ such that for any $a\in\A$ one has
\begin{equation}
  \label{eq:29}
  \varphi_1 (a) =\int_{\pa} \omega(a)\, \de\mu_1(\omega), \quad   \varphi_2 (a) =\int_{\pa} \omega(a)\, \de\mu_2(\omega)
\end{equation}
for two given states $\varphi_1, \varphi_2$. However it is not yet clear that (\ref{eq:33}) is a distance on $\sa$.
\medskip

In \cite{Martinetti:2012fk},  we explored another way, consisting in viewing $\A$ as an
``noncommutative algebra of functions'' on $\pa$,
\begin{equation}
  \label{eq:3mille}
  a(\omega) := \omega(a) \quad \forall \omega\in\pa, \, a\in\A;
\end{equation}
and define the set of ``$d_D$-Lipschitz
noncommutative functions''   in analogy with
\eqref{eq:6} as
\begin{equation}
 \text{Lip}_D(\A) := \left\{
 a\in\A \,,\; \abs{a(\omega_1) - a(\omega_2)} \leq d_D(\omega_1, \omega_2) \quad
  \forall \,\omega_1, \omega_2\in\pa\right\}. 
\label{eq:5}
\end{equation}
By mimicking~\eqref{eq:1} we then defines for any $\varphi_1, \varphi_2\in \sa$
  \begin{equation}
    \label{eq:2}
    W_D(\varphi_1, \varphi_2) := \sup_{a\in\lda}\abs{\varphi_1(a)
        -\varphi_2(a)}.
  \end{equation}
\begin{prop}
\cite[Prop. 3.1]{Martinetti:2012fk}
\label{propmkncg}
$W_D$ is a distance, possibly infinite, on $\sa$. Moreover for any $\varphi_1, \varphi_2\in\sa$,
  \begin{equation}
\label{inclusionn}
d_D(\varphi_1, \varphi_2) \leq W_D(\varphi_1, \varphi_2).
\end{equation}
The equation above is an equality on the set of convex linear combinations
\begin{equation}
\varphi_\lambda := \lambda\, \omega_1 + (1-\lambda) \,\omega_2\label{eq:156}
\end{equation}
of any two given pure states $\omega_1, \omega_2$: namely
for any $\lambda, \tilde\lambda\in [0,1]$ one has
\begin{equation}
  \label{eq:35}
  d_D(\varphi_\lambda, \varphi_{\tilde\lambda}) = \abs{\lambda -
    \tilde\lambda}\,d_D(\omega_1, \omega_2) = W_D(\varphi_\lambda, \varphi_{\tilde\lambda}).
\end{equation}
 \end{prop}

The difference between $W_D$ and $d_D$ is entirely contai\-ned in the difference between
the $D$-Lipschitz ball~\eqref{eq:14} and
$\lda$ defined in (\ref{eq:5}). In the commutative case
$\A=C_0(\M)$, these two notions of Lipschitz functions
coincide with the usual one, so that $d_{\ds} = W_{\ds}$. In the
noncommutative case, they coincide on some easy low dimensional
examples, like for $\A=M_2(\C)$, but there are indications that this is not true in general
\cite[\S 7]{rieffel2003}.
\bigskip

To conclude, let us mention another direction of research still
largely unexplored:
generalizing to the noncommutative realm the Wasserstein distance
$W_p$ of
order $p\geq 2$ \eqref{eq:138}. The only attempt we are aware of is that of \cite{cite-key},
where one proposes a noncommutative version of $W_p$ based on the posets of
commutative sub-algebras of a noncommutative 
$C^*$-algebra. 
\vspace{2truecm}
\begin{center}
\line(1,0){200}
\end{center}
\newpage

\section*{Notations}

Given $z\in\C$, we denote $\bar z$ its
 conjugate, $|z|$ its module and ${\mathfrak R}(z)$ its real part.
\smallskip

 Given an involutive algebra $\A$, the
adjoint of an element $a\in\A$ is $a^*$.  A $C^*$-algebra is an
associative and involutive algebra $\A$, equipped with a norm $||\cdot||$
in which it is complete, and such that for any $a\in\A$ one has
\begin{equation}
||a^*
a|| =||a||^2.
\label{eq:139}
\end{equation}
It is \emph{unital} if it contains a unit, that is an
element ${\bf 1}\in\A$ such that
\begin{equation}
{\bf 1} a = a{\bf 1} = a\quad \forall a\in\A.\label{eq:140}
\end{equation}

 Most of the time we identify an element $a$
of the algebra $\A$ with its representation $\pi(a)$ as bounded operator on
some Hilbert space $\HH$. Unless otherwise specified, representations
are always faithful and non-degenerate. In particular,
     if $\A$ is unital with unit ${\bf 1}$, this guarantees that
     $\pi({\bf 1})$ is the identity $\I$ of ${\mathcal B}(\HH)$.
\medskip

The adjoint action of a unitary $u\in\A$ is $(\text{Ad} u)(a) := u
au^*$, for any $a\in\A$.
\medskip

We denote by $\sa$ and $\pa$
the space of states and of pure states of $\A$. We usually denote
a state by $\varphi$, with suitable decorations $\varphi', \varphi_0,
\varphi_1, ...$ if needed. A faithful state is a state $\varphi$ such state
$\varphi(a^*a)=0$ iff $a=0$. Pure states are usually denoted by
$\omega$ in the noncommutative case, and $\delta$ in the commutative case, with suitable decorations.
\medskip

 Given two operators $A, B$
acting on an Hilbert space $\HH$,  the bracket $[A, B]= AB - BA$ is their
commutator. ${\mathcal B}(\HH)$ denotes the space of bounded operators on $\HH$, and
$\I$ is the identity operator. Unless otherwise specified, the norm $||\cdot||$ is the
operator norm coming from the action on $\HH$, that is
\begin{equation}
  \label{eq:132}
  ||A|| = \sup_{\psi\in\HH} \frac{||A\psi||_{\HH}}{||\psi||_{\HH}}
\end{equation}
where
\begin{equation}
||\psi||_{\HH} = \sqrt{\langle \psi, \psi\rangle}
\label{eq:136}
\end{equation}
is the $L^2$-norm on $\HH$, with $\langle \cdot, \cdot \rangle$ the inner
product on $\HH$. We omit the index $_\HH$ and it should be
clear from the context if one deals with the operator or the
$L^2$-norm. 

A vector in $\HH$ is usually denoted by the greek letter $\psi, \zeta$
or $\xi$. Its components in a given orthonormal basis are the complex
numbers $\psi_i,
\zeta_i, \xi_i$ with $i =1, ..., \text{dim} \,\HH$. The dual vector is
$\bar \psi$, with components $\bar\psi_i$.  
 The canonical basis of $\C^N$ is $\left\{e_{ij}\right\}_{i,j =1, ...,N}$, that is $e_{ij}$ is the matrix with null entries, except $1$
at the $i^{\text{th}}$ line, $j^{\text{th}}$ column.
\medskip

A spectral triple $(\A,\HH,D)$ is the datum of a (non necessarily
commutative) involutive algebra $\A$, acting faithfully on an Hilbert space
$\HH$ via a representation $\pi$, together with a (non necessarily bounded) operator $D$ on $\HH$
such that $D-\lambda\I$ is compact for any $\lambda$ in the
resolvent set of $D$, and $[D, \pi(a)]$ is a
bounded operator for any
$a\in\A$. A spectral triple is unital when the algebra $\A$ is unital
and the representation $\pi$ is non-degenerate. It is graded if there
exists a grading $\Gamma$ of $\HH$ (that is $\Gamma=\Gamma^*$ and
$\Gamma^2=\I$) which commutes  with any $\pi(a)$ and anticommutes
with $D$. The $D$-Lipschitz ball
of $\A$ is 
\begin{equation}
  \label{eq:145}
  L_D(\A):= \left\{ a\in\A, ||[D, a]||\leq 1\right\}.
\end{equation}

$D$ denotes the (generalized) Dirac operator of an arbitrary spectral
triple. $\ds$ is the usual Dirac operator of a spin
manifold. $\gamma^\mu$ are the Dirac matrices, $\sigma^\mu$ the Pauli matrices.
\medskip

We call ``distance'' a function that verifies all the usual properties of a
 distance,  except that we do not assume it is necessarily finite. 
$d$ is the spectral distance \eqref{eq:-1}, $d_A$ the fluctuated distance
defined in \S~\ref{subriemannian}, $d_h$ the horizontal distance in
sub-Riemannian geometry.

A path metric space is a metric space $(\X,\tilde d)$ such that between any
two points $x,y\in\X$ there exists a continuous curve $c: [0,1]\to \X$
with $c(0)=x, c(1)=y$ and such that
\begin{equation}
  \label{eq:157}
  \tilde d(c(s), c(t))= |s-t|\, \tilde d(x,y) \quad \forall s,t \in [0,1].
\end{equation}

\medskip
The Lipschitz norm of a function $f$ on a Riemannian manifold $\M$ with
geodesic distance $d_{\text{geo}}$ is
\begin{equation}
  \label{eq:137}
  ||f||_{\text{Lip}}:=\sup_{x,y \in \M}\frac{|f(x) -
    f(y)|}{d_{\text{geo}(x, y)|}}.
\end{equation}

The algebra of $n$-dimensional complex matrices is $M_n(\C)$. The
algebra of quaternions is $\mathbb H$.
\medskip

 \providecommand{\bysame}{\leavevmode\hbox to3em{\hrulefill}\thinspace}
\providecommand{\MR}{\relax\ifhmode\unskip\space\fi MR }
\providecommand{\MRhref}[2]{%
  \href{http://www.ams.org/mathscinet-getitem?mr=#1}{#2}
}
\providecommand{\href}[2]{#2}

\end{document}